\documentclass[a4paper]{article}

\usepackage[utf8]{inputenc}
\usepackage{a4wide}
\usepackage{amsmath,amssymb,amsthm}
\usepackage{enumerate}
\usepackage{hyperref}
\usepackage{stmaryrd}
\usepackage{subfig}

\usepackage{tikz}
\usetikzlibrary{positioning,patterns}

\tikzset{gate/.style={draw, circle, inner sep = 1.5pt, minimum size = 17pt}}
\tikzset{aux-gate/.style={gate, minimum size = 5pt}}
\tikzset{input/.style={gate, minimum size = 0pt, inner sep = 2pt, rectangle, rounded corners = 2pt}}
\tikzstyle{bot} = [fill=none, draw=gray, densely dotted, thin]
\tikzstyle{marked} = [fill=gray!30]
\tikzstyle{every path}=[semithick]

\newif\ifmac
\mactrue

\newcommand{\B}{\mathbb{B}}
\newcommand{\C}{\mathcal{C}}

\newcommand{\D}{\mathcal{D}}
\newcommand{\G}{\mathcal{G}}
\newcommand{\N}{\mathbb{N}}
\newcommand{\Z}{\mathbb{Z}}

\newcommand{\DLOGTIME}{\mathsf{DLOGTIME}}
\newcommand{\DSPACE}{\mathsf{DSPACE}}
\newcommand{\DET}{\mathsf{DET}}
\newcommand{\coRNC}{\mathsf{coRNC}}
\newcommand{\Log}{\mathsf{L}}
\newcommand{\LOGCFL}{\mathsf{LOGCFL}}
\newcommand{\NL}{\mathsf{NL}}
\newcommand{\NP}{\mathsf{NP}}
\newcommand{\Ptime}{\mathsf{P}}
\newcommand{\PSPACE}{\mathsf{PSPACE}}
\newcommand{\AC}{\mathsf{AC}}
\newcommand{\NC}{\mathsf{NC}}

\newcommand{\CEP}{\mathsf{CEP}}
\newcommand{\aR}{R_+}
\newcommand{\mR}{R_{{\scriptscriptstyle \bullet}}}

\newcommand{\UCEP}{\mathsf{UCEP}}
\newcommand{\rank}{\mathsf{rank}}
\newcommand{\rhs}{\mathsf{rhs}}
\newcommand{\supp}{\mathsf{supp}}
\newcommand{\val}{\mathsf{val}}

\newcommand{\fval}[1]{{\llbracket #1 \rrbracket}}

\newcommand{\problem}[2]{\smallskip\noindent{\bf Input:} #1\\{\bf Question:} #2\smallskip}
\newcommand{\outputproblem}[2]{\smallskip\noindent{\bf Input:} #1\\{\bf Output:} #2\smallskip}

\newtheorem{theorem}{Theorem}
\numberwithin{theorem}{section}
\newtheorem{corollary}[theorem]{Corollary}
\newtheorem{definition}[theorem]{Definition}
\newtheorem{lemma}[theorem]{Lemma}
\newtheorem{example}[theorem]{Example}
\newtheorem{proposition}[theorem]{Proposition}
\newtheorem*{inv}{Invariant}

\title{Circuit Evaluation for Finite Semirings}

\author{Moses Ganardi, Danny Hucke, Daniel K\"onig, Markus Lohrey}

\date{
	\small University of Siegen, Germany \\	
	\texttt{\{ganardi,hucke,koenig,lohrey\}@eti.uni-siegen.de}
}

\begin{document}

\maketitle

\begin{abstract}
The computational complexity of the circuit evaluation problem for finite semirings is considered, where
semirings are not assumed to have an additive or multiplicative identity. The following dichotomy is shown:
If a finite semiring is such that (i) the multiplicative semigroup is solvable and 
(ii) it does not contain a subsemiring with an additive identity $0$ and a multiplicative identity $1 \neq 0$,
then the circuit evaluation problem for the semiring is in $\DET \subseteq \NC^2$. In all other cases,
the circuit evaluation problem is $\Ptime$-complete.
\end{abstract}

\section{Introduction}

Circuit evaluation problems are among the most well-studied computational problems in complexity theory.
In its most general formulation, one has an algebraic structure $\mathcal{A}=(A, f_1, \ldots, f_k)$, where the $f_i$ are 
mappings $f_i : A^{n_i} \to A$. A circuit over the structure $\mathcal{A}$ is a directed acyclic graph (dag) where every inner node 
is labelled with one of the operations $f_i$ and has exactly $n_i$ incoming edges that are linearly ordered. 
The leaf nodes of the dag are labelled with elements of $A$ (for this, one needs a suitable finite representation of elements from $A$),
and there is a distinguished output node. The task is to evaluate this dag in the natural way, and to return the value of the output node.

In his seminal paper \cite{Lad75},
Ladner  proved that the circuit evaluation problem for the Boolean semiring $\mathbb{B}_2 = (\{0,1\}, \vee, \wedge)$ is 
$\Ptime$-complete. This result marks a cornerstone in the theory of $\Ptime$-completeness \cite{GrHoRu95}, and motivated the investigation
of circuit evaluation problems for other algebraic structures. A large part of the literature is focused on arithmetic (semi)rings like 
$(\mathbb{Z}, +, \cdot)$, $(\mathbb{N}, +, \cdot)$ or the max-plus semiring $(\mathbb{Z} \cup \{-\infty\}, \max, +)$ \cite{AllenderBKM09,AllenderJMV98,Kosaraju90,MillerRK88,MillerT99,ValiantSBR83}. 
These papers mainly consider circuits of polynomial formal degree. For commutative semirings, circuits of polynomial formal degree
can be restructured into an equivalent (unbounded fan-in) circuit of polynomial size and
logarithmic depth \cite{ValiantSBR83}.
This result leads to $\NC$-algorithms for evaluating polynomial degree circuits over commutative semirings \cite{MillerRK88,MillerT99}.
Over non-commutative semirings, circuits of polynomial formal degree do in general not allow a restructuring into 
circuits of logarithmic depth \cite{Kosaraju90}.

In \cite{MillerT99} it was shown that also for finite non-commutative semirings circuit evaluation is in $\NC$ for circuits of polynomial formal degree.
On the other hand, the authors are not aware of any $\NC$-algorithms for evaluating general (exponential degree) circuits  over semirings. 
The lack of such algorithms is probably due to Ladner's result, which seems to exclude any efficient parallel algorithms. 
On the other hand, in the context of semigroups, there exist $\NC$-algorithms for circuit evaluation.
In \cite{BeMcPeTh97}, the following dichotomy result was shown for finite semigroups: If the finite semigroup is solvable (meaning that every subgroup is a solvable group),
then circuit evaluation is in $\NC$ (in fact, in $\DET$, which is the class of all problems that are $\AC^0$-reducible to the computation of an 
integer determinant \cite{Coo85,Coo12}), otherwise circuit evaluation is $\Ptime$-complete. 

In this paper, we extend the work of  \cite{BeMcPeTh97} from finite semigroups to finite semirings. On first sight, it seems 
again that Ladner's result excludes efficient parallel algorithms: It is not hard to show that if the finite semiring has an additive identity $0$
and a multiplicative identity $1 \neq 0$ (where $0$ is not necessarily absorbing with respect to multiplication), then circuit evaluation is $\Ptime$-complete,
see Lemma~\ref{lemma-P-hard}. Therefore, we take the most general reasonable definition of semirings: A semiring is a structure
$(R,+,\cdot)$, where $(R,+)$ is a commutative semigroup, $(R,\cdot)$ is a semigroup, and $\cdot$ distributes (on the left and right) over $+$.
In particular, we neither require the existence of a $0$ nor a $1$. Our main result states that in this general setting
there are only two obstacles to efficient parallel circuit evaluation: non-solvability of the multiplicative structure and the existence of a zero 
and a one (different from the zero) in a subsemiring. More precisely, we show the following two results, where a semiring is called 
$\{0,1\}$-free if there exists no subsemiring in which an additive identity $0$ and a multiplicative identity $1 \neq 0$ exist:
\begin{enumerate}[(1)]
\item If a finite semiring is not $\{0,1\}$-free, then the circuit evaluation problem is $\Ptime$-complete.
\item If a finite semiring  $(R,+,\cdot)$ is $\{0,1\}$-free, then the circuit evaluation problem for $(R,+,\cdot)$ can
be solved with $\AC^0$-circuits that are equipped with oracle gates for (a) graph reachability, (b) the circuit evaluation problem for 
the commutative semigroup $(R,+)$ and (c) the circuit evaluation problem for 
the semigroup $(R,\cdot)$.
\end{enumerate}
Together with the dichotomy result from \cite{BeMcPeTh97} (and the fact that commutative semigroups are solvable)
we get the following result: For every finite semiring $(R,+,\cdot)$, the circuit
evaluation problem is in $\NC$ (in fact, in $\DET$) if $(R,\cdot)$ is solvable and $(R,+,\cdot)$ is $\{0,1\}$-free.
Moreover, if one of these conditions fails,
then circuit evaluation is $\Ptime$-complete.    

The hard part of the proof is to show the above statement (2).
We will proceed in two steps. In the first step we reduce the circuit evaluation problem for a finite semiring $R$ to the 
evaluation of a so called type admitting circuit. This is a circuit where every gate evaluates to an element
of the form $e a f$, where $e$ and $f$ are  multiplicative idempotents of $R$. Moreover, these idempotents $e$ and $f$
have to satisfy a certain compatibility condition that will be expressed by a so called type function. 
In a second step, we present a parallel evaluation algorithm for type admitting circuits. Only for this second step we need
the assumption that the semiring is $\{0,1\}$-free.

In Section~\ref{sec-intersection} we present an application of our main result for circuit evaluation to formal language theory.
We consider the intersection non-emptiness problem for a given context-free language and a fixed regular language $L$. 
If the context-free language is given by an arbitrary context-free grammar, then we show that the intersection non-emptiness problem
is $\Ptime$-complete as long as $L$ is not empty (Theorem~\ref{thm-p-complete-intersection}). It turns out that the reason
for this is non-productivity of nonterminals. We therefore consider a restricted version of the intersection non-emptiness problem, where 
every nonterminal of the input context-free grammar must be productive. To avoid a promise problem (testing productivity
of a nonterminal is $\Ptime$-complete), we in addition provide a witness of productivity for every nonterminal. This witness
consists of exactly one production $A \to w$ for every nonterminal of $A$ 
such that the set of all the set of all selected productions is an acyclic grammar $\mathcal{H}$. This ensures that
$\mathcal{H}$ derives for every nonterminal $A$ exactly one string that is a witness of the productivity of $A$. 
We then show that this restricted version of the intersection non-emptiness problem with the fixed regular language $L$ is equivalent
(with respect to constant depth reductions) to the circuit evaluation problem for a certain finite semiring that is derived from
the syntactic monoid of the regular language $L$.

\paragraph*{Further related work}
We mentioned already the existing work on circuit evaluation for infinite semirings.  The question whether a given circuit over a polynomial ring
evaluates to the zero polynomial is also  known as polynomial identity testing. For polynomial rings over $\Z$ or $\Z_n$ ($n \geq 2$), polynomial
identity testing has a co-randomized polynomial time algorithm \cite{AgrawalB03,IbMo83}. Moreover, the question, whether a deterministic polynomial time algorithm exists
is tightly related to lower bounds in complexity theory, see \cite{ShpilkaY10} for a survey.
For infinite groups, the circuit evaluation problem is also known as the compressed word problem \cite{Loh14}. In the context of parallel algorithms,
it is interesting to note that the third and forth author recently proved that the circuit evaluation problem for finitely generated (but infinite) nilpotent
groups belongs to $\DET$ \cite{KonigL15}. 
For finite non-associative groupoids, the complexity of circuit evaluation was studied in \cite{MooreTLBD00},
and some of the results from \cite{BeMcPeTh97} for semigroups were generalized to the non-associative setting.
In \cite{BeaudryH07}, the problem of evaluating tensor circuits is studied. The complexity of this problem is quite high: Whether a given tensor
circuit over the Boolean semiring evaluates to the $(1 \times 1)$-matrix $(0)$ is complete for nondeterministic exponential time.
Finally, let us mention the papers \cite{McKenzieW07,Travers06}, where circuit evaluation problems are studied 
for the power set structures $(2^{\N}, +, \cdot, \cup, \cap, \overline{\phantom{a}})$ and $(2^{\Z}, +, \cdot, \cup, \cap, \overline{\phantom{a}})$,
where $+$ and $\cdot$ are evaluated on sets via $A \circ B = \{a \circ b \mid a \in A, b \in B\}$.

A variant of our intersection non-emptiness problem was studied in \cite{RubtsovV15}. There, a context-free language $L$ is fixed,
a (deterministic or non-deterministic) finite automaton $\mathcal{A}$ is the input, and the question is, whether $L \cap L(\mathcal{A}) = \emptyset$ holds. The authors present large classes of context-free languages such that for each member the intersection non-emptiness problem
with a given regular language (specified by a non-deterministic automaton) is $\Ptime$-complete (resp., $\NL$-complete).

\section{Computational complexity}

For background in complexity theory the reader might consult \cite{AroBar09}.
We assume that the reader is familiar with the complexity classes 
$\NL$ (non-deter\-ministic logspace) and
 $\Ptime$ (deterministic polynomial time). 
 A function is logspace-com\-putable if it can be computed by a deterministic Turing-machine 
 with a logspace-bounded work tape, a read-only input tape, and a write-only output tape.
 Note that the logarithmic space bound only applies to the work tape.
$\Ptime$-hardness will refer to logspace reductions.

We use standard definitions concerning circuit complexity, see e.g. \cite{Vol99}. 
We only consider polynomially bounded families $(C_n)_{n \geq 0}$ of Boolean circuits, where 
the number of gates of $C_n$ is bounded by a polynomial $p(n)$. For such a family, 
gates of $C_n$ can  be encoded with bit strings of length $O(\log n)$.
The family $(C_n)_{n \geq 0}$ is $\DLOGTIME$-uniform, if  
for given binary coded gates  $u, v$ of $C_n$, one can (i) compute the type of gate $u$ in time
$O(\log n)$ and (ii) check in  time $O(\log n)$ whether $u$ is an input gate for $v$.
Note that the time bound $O(\log n)$ is linear in the input length $|u|+|v|$.
All circuit families in this paper are implicitly assumed to be $\DLOGTIME$-uniform.
 We will consider the class $\AC^0$ of all problems
that can be recognized by a polynomial size circuit family of constant depth built up from NOT-gates (which have fan-in one) and AND- and OR-gates of unbounded fan-in.
The class $\NC^k$ ($k \geq 1$) is defined by polynomial size circuit families
of depth $O(\log^k n)$ that use NOT-gates, and AND- and OR-gates of fan-in two.
One defines $\NC = \bigcup_{k \geq 1} \NC^k$.
The above language classes can be easily generalized to classes of functions by allowing circuits with several
output gates. Of course, this only allows to compute functions $f : \{0,1\}^* \to \{0,1\}^*$ such that $|f(x)| = |f(y)|$
whenever $|x| = |y|$. If this condition is not satisfied, one has to consider a suitably padded version of $f$.

We use the standard notion of constant depth Turing-reducibility:
For functions $f_1, \ldots, f_k$ let $\AC^0(f_1,\ldots, f_k)$ be the class of all functions that can be computed with 
a  polynomial size circuit family of constant depth that uses NOT-gates and unbounded fan-in AND-gates, OR-gates, and 
$f_i$-oracle gates ($1 \leq i \leq k$). Here, an $f_i$-oracle gate receives an ordered tuple of inputs $x_1, x_2, \ldots, x_n$ and outputs 
the bits of $f_i(x_1 x_2 \cdots x_n)$. 
By taking the characteristic function of a language, 
we can also allow a language $L_i  \subseteq \{0,1\}^*$ in place of $f_i$.
Note that the function class $\AC^0(f_1,\ldots, f_k)$ is closed under composition (since the composition of two
$\AC^0$-circuits is again an $\AC^0$-circuit). 
We write $\AC^0(\NL, f_1, \ldots, f_k)$ for $\AC^0(\mathsf{GAP}, f_1, \ldots, f_k)$, where
$\mathsf{GAP}$ is the $\NL$-complete graph accessibility problem.
The class $\AC^0(\NL)$ is studied in \cite{AlvarezBJ91}. It has several alternative characterizations and can
be viewed as a nondeterministic version of functional logspace. 
As remarked in  \cite{AlvarezBJ91}, the restriction of $\AC^0(\NL)$ to $0$-$1$ functions is $\NL$.
Clearly, every logspace-computable function belongs to $\AC^0(\NL)$: The $\NL$-oracle can be used to directly
compute the output bits of a logspace-computable function.

Let $\DET = \AC^0(\mathsf{det})$, where $\mathsf{det}$ is the function that maps a binary encoded
integer matrix to the binary encoding of its determinant, see \cite{Coo85}. Actually, Cook defined
$\DET$ as $\NC^1(\mathsf{det})$ \cite{Coo85}, but the above definition via $\AC^0$-circuits seems
to be more natural. For instance, it implies that $\DET$ is equal to the $\#\Log$-hierarchy,
see also the discussion in \cite{Coo12}. 

We defined $\DET$ as a function class, but the definition can be extended to languages by considering their 
characteristic functions. It is well known that $\NL \subseteq \DET \subseteq \NC^2$ \cite{Coo12}. 
From $\NL \subseteq \DET$, it follows easily that $\AC^0(\NL, f_1, \ldots, f_k) \subseteq \DET$ whenever
$f_1, \ldots, f_k \in \DET$.

\section{Algebraic structures, semigroups, and semirings}

An {\em algebraic structure} $\mathcal{A} = (A,f_1, \dots, f_k)$
consists of a non-empty {\em domain} $A$ and operations $f_i: A^{n_i} \to A$
for $1 \le i \le k$. 
We often identify the domain with the structure, if it is clear from the context.
A {\em substructure} of $\mathcal{A}$ is a subset $B \subseteq A$ that 
closed under each of the operations $f_i$. We identify $B$ with the structure
$(B, g_1, \ldots, g_k)$, where $g_i : B^{n_i} \to B$ is the restriction of $f_i$ to $B^{n_i}$
for all $1 \le i \le k$.

We mainly deal with semigroups and semirings. In the following two subsection
we present the necessary background. 
For further details on semigroup theory (resp., semiring theory) see
\cite{RhSt} (resp., \cite{golan1999semirings}).

\subsection{Semigroups} \label{sec-semigroups}

A {\em semigroup} $(S,\circ)$(or just $S$)  is an algebraic structure with a single associative binary operation.
We usually write $st$ for $s \circ t$.
If $st = ts$ for all $s,t \in S$, we call $S$ {\em commutative}.	
A set $I \subseteq S$ is called a {\em semigroup ideal} if for all $s \in S$, $a \in I$ we have $sa, as \in I$.
An element $e \in S$ is called {\em idempotent} if $ee = e$. It is well-known that
for every finite semigroup $S$ and $s \in S$ there exists an $n \geq 1$ such that $s^n$ is idempotent. 
In particular, every finite semigroup contains an idempotent element.
By taking the smallest common multiple of all these $n$, one obtains an $\omega \geq 1$
such that $s^\omega$ is idempotent for all $s \in S$.
The set of all idempotents of $S$ is denoted with $E(S)$.
If $S$ is finite, then $S E(S) S = S^n$ where $n = |S|$. Moreover, $S^n = S^m$ for all $m \geq n$.
For a set $\Sigma$, the {\em free semigroup} generated by $\Sigma$ is the set $\Sigma^+$ of all finite non-empty words over $\Sigma$
together with the operation of concatenation.

A semigroup $M$ with an identity element $1 \in M$, i.e. $1 m = m 1 = m$ for all $m \in M$,
is called a {\em monoid}. With $S^1$ we denote the monoid that is obtained from a semigroup $S$ by adding a fresh element $1$, which
becomes the identity element of $S^1$. Thus, we extend the multiplication to $S^1 = S \cup \{1\}$
by setting $1s = s1 = s$ for all $s \in S \cup \{1\}$. 
In case $M$ is monoid and $N$ is a submonoid of $M$, we do not require that the identity element of $N$ is the identity
element of $M$.  But, clearly, the identity element of the submonoid $N$ must be an idempotent element of $M$.
In fact, for every semigroup $S$ and every idempotent $e \in E(S)$, the set $eSe = \{ ese \mid s \in S\}$
is a submonoid of $S$ with identity $e$, which is also called a {\em local submonoid} of $S$.
The local submonoid $eSe$ is the maximal submonoid of $S$ whose identity element is $e$. 
A semigroup $S$ is {\em aperiodic} if every subgroup of $S$ is trivial.
A semigroup $S$ is {\em solvable} if every subgroup $G$ of $S$ is a solvable group, i.e.,  
for the series defined by $G_0 = G$ and $G_{i+1} = [G_i,G_i]$ (the commutator subgroup of $G_i$)
there exists an $i \geq 0$ with $G_i=1$.
Since Abelian groups are solvable, every commutative semigroup is solvable.

\subsection{Semirings}

A {\em semiring} $(R,+,\cdot)$ consists of a non-empty set $R$ with two operations $+$ and $\cdot$
such that $(R,+)$ is a commutative semigroup, $(R,\cdot)$ is a semigroup,
and $\cdot$ left- and right-distributes over $+$, i.e., $a \cdot (b+c) = ab + ac$ and 
$(b+c) \cdot a = ba + ca$ (as usual, we write $ab$ for $a \cdot b$). 
Note that we neither require the existence of an additive identity $0$ nor the 
existence of a multiplicative identity $1$.
We denote with $\aR = (R,+)$ the additive semigroup of $R$ and 
with $\mR = (R,\cdot)$ the multiplicative semigroup of $R$.
For $n \geq 1$ and $r \in R$ we write
$n \cdot r$ or just $n r$ for $r + \cdots + r$, where $r$ is added $n$ times. 
With $E(R)$ we denote the set of multiplicative idempotents of $R$,
i.e., those $e \in R$ with $e^2 = e$.
Note that for every multiplicative idempotent $e \in E(R)$, $eRe$ is a subsemiring of $R$ in which the multiplicative
structure is a monoid. 
For a non-empty subset $T \subseteq R$ we denote by $\langle T \rangle$ the subsemiring generated by $T$,
i.e. the smallest set containing $T$ which is closed under addition and multiplication.
An {\em ideal} of $R$ is a subset $I \subseteq R$ such that for all $a,b \in I$, $s \in R$ we have
$a+b,sa, as \in I$. Clearly, every ideal is a subsemiring. Let 
$\B_2 = (\{0,1\},\vee,\wedge)$ be the {\em Boolean semiring}.

 For a given non-empty set $\Sigma$, the {\em free semiring} $\N[\Sigma]$ generated by $\Sigma$ consists of all mappings 
 $f : \Sigma^+ \to \N$ such that $\supp(f) := \{ w \in \Sigma^+ \mid f(w) \neq 0\}$ is finite and non-empty.
 Addition is defined pointwise, i.e., $(f+g)(w) = f(w) + g(w)$, and multiplication is defined by the convolution:
 $(f \cdot g)(w) = \sum_{w = uv} f(u) \cdot g(v)$, where the sum is taken over all factorizations $w = uv$ with
 $u,v \in \Sigma^+$. 
 We view an element $f \in \N[\Sigma]$ as a non-commutative polynomial $\sum_{w \in \supp(f)} f(w) \cdot w$.
 Then addition (resp.~multiplication) in $\N[\Sigma]$ corresponds to addition (resp.~multiplication) of non-commutative
 polynomials. Words $w \in \supp(f)$ are also called {\em monomials} of $f$. A word $w \in \Sigma^+$ is identified
 with the non-commutative polynomial $1 \cdot w$, i.e., the mapping $f$ with $\supp(f) = \{w\}$ and $f(w) = 1$.
 For every semiring $R$ which is generated by $\Sigma$ there exists a canonical surjective homomorphism from $\N[\Sigma]$
 to $R$ which evaluates non-commutative polynomials over $\Sigma$. Since a semiring is not assumed to have 
 a multiplicative identity (resp., additive identity), we have to exclude the empty word from $\supp(f)$ for every $f \in \N[\Sigma]$
 (resp., exclude the mapping $f$ with $\supp(f) = \emptyset$ from $\N[\Sigma]$).
 
A crucial definition in this paper is that of a {\em $\{0,1\}$-free} semiring. This is a semiring $R$ which does {\em not}
contain a subsemiring $T$ with an additive identity $0$ and 
a multiplicative identity $1 \neq 0$. Note that in such a semiring $T$ we do not require that $0$ is absorbing, i.e., $a \cdot 0 = 0 \cdot a = 0$ for all $a \in T$.
The class of $\{0,1\}$-free {\em finite} semirings has several characterizations:

\begin{lemma} \label{lemma-0-1-free}
	For a finite semiring $R$, the following are equivalent:
	\begin{enumerate}
		\item $R$ is not $\{0,1\}$-free.
		\item $R$ contains $\B_2$ or the ring $\Z_d$ for some $d \ge 2$ as a subsemiring.
		\item $R$ is divided by $\B_2$ or $\Z_d$ for some $d \ge 2$ (i.e., $\B_2$ or $\Z_d$ is a homomorphic image of 
		a subsemiring of $R$).	
		\item There exist elements $0,1 \in R$ such that $0 \neq 1$, $0+0=0$, $0+1=1$, $0 \cdot 1 = 1 \cdot 0 = 0 \cdot 0 = 0$, and $1 \cdot 1 = 1$ (but $1+1 \neq 1$ is possible).
	\end{enumerate}
\end{lemma}

\begin{proof}
($1 \Rightarrow 2$): Let $T$ be a subsemiring of $R$ which has a zero element $0$ and a one element $1 \neq 0$.
Note that $0 \cdot 0 = 0 \cdot 0 + 0 = 0 \cdot 0 + 1 \cdot 0 = (0+1) \cdot 0 = 1 \cdot 0 = 0$.
Let $T' = \{0\} \cup \{k \cdot 1 \mid k \in \N \}$, which is the subsemiring generated by these elements.
It is isomorphic to some semiring $B(k,d)$ ($k \geq 0$, $d \geq 1$), which is the semiring $(\N,+,\cdot)$
modulo the congruence relation $\sim$ defined by
$i \sim j$ if $0 \le i = j \le k-1$ or ($i,j \ge k$ and $d$ divides $i-j$). Since $0 \neq 1$, we have $(k,d) \neq (0,1)$.
If $k = 0$, then $B(0,d)$ is isomorphic to $\mathbb{Z}_d$ for $d \geq 2$.
If $k \ge 1$, then choose $a \geq k$ such that $d$ divides $a$, for example $a = dk$.
Then $\{0,a \cdot 1\}$ is a subsemiring isomorphic to the Boolean semiring $\B_2$. 

($2 \Rightarrow 3$): This implication  is trivial. 

($3 \Rightarrow 4$): Assume that $\varphi : T \to T'$ is a homomorphism from a subsemiring $T$ of $R$ to
$T'$, where the latter is $\B_2$ or $\Z_d$ with $d \geq 2$. In particular, there exist $0,1 \in T'$ with 
$0 \neq 1$, $0+0=0$, $0+1=1$, $0 \cdot 1 = 1 \cdot 0 = 0 \cdot 0 = 0$, and $1 \cdot 1 = 1$.
Let $n \geq 1$ be such that $n \cdot x$ is additively idempotent and $x^n$ is multiplicatively idempotent
for all $x \in R$. Then $n \cdot x^n$ is  is additively and multiplicatively idempotent
for all $x \in R$. Let $a,e \in T$ be such that $\varphi(a) = 0$ and $\varphi(e) = 1$.
We can replace $a$ by $n \cdot a^n$ and $e$ by  $e^n$. Then, $a+a = aa = a$
and $ee=e$.
For $a' = n \cdot (eae)^n$ we have $\varphi(a') = 0$ and $a'e = ea' = a'+a' = a'a' = a'$.
For $e' = a'+e$ we have $\varphi(e') = 1$ (hence, $a' \neq e'$) 
and $e' e' = a' a' + a' e + e a' + ee = a'+e = e'$,  $a'+e'= a'+a'+e = e'$.
Furthermore, we have $a'e' = a' (a'+e) = a' + a'e = a'$ and similarly $e'a' = a'$.
Hence, $a'$ and $e'$ satisfy all equations from point 4.

($4 \Rightarrow 1$): Assume that there exist 
elements $0,1 \in R$ such that $0 \neq 1$, $0+0=0$, $0+1=1$, $0 \cdot 1 = 1 \cdot 0 = 0 \cdot 0 = 0$, and $1 \cdot 1 = 1$.
Consider the subsemiring generated by $\{0,1\}$, which is 
$\{0\} \cup \{n \cdot 1 \mid n \geq 1 \}$. By the above identities  $0$ (resp., $1$) is an additive (resp., multiplicative) identity
in this subsemiring.
%($1 Rightarrow 4$): First assume that $R$ is not $\{0,1\}$-free. Then there exists a subsemiring $T$ with element $0,1 \in T$, $0 \neq 1$, such 
%that $0$ (resp., $1$) is an additive (resp., multiplicative) identity in $T$.
%In particular, we have $0+0=0$, $0+1=1$, $0 \cdot 1 = 1 \cdot 0 =  0$, and $1 \cdot 1 = 1$.
%Moreover, we have
%$0 \cdot 0 = 0 \cdot 0 + 0 = 0 \cdot 0 + 1 \cdot 0 = (0+1) \cdot 0 = 1 \cdot 0 = 0$.
\end{proof}
As a consequence of Lemma~\ref{lemma-0-1-free} (point 4), one can check in time $O(n^2)$ for a semiring
of size $n$ whether it is $\{0,1\}$-free. We will not need this fact, since in our setting the semiring will be always fixed, i.e., not part of the input.
Moreover, the class of all $\{0,1\}$-free semirings is closed under taking subsemirings (this is trivial) and taking homomorphic images (by point 3).
Finally, the class of $\{0,1\}$-free semirings is also closed under direct products. To see this, assume that $R \times R'$ 
is not $\{0,1\}$-free. Hence, there exists a subsemiring $T$ of $R \times R'$ with an additive zero $(0, 0')$ and a multiplicative one $(1, 1') \neq (0,0')$.
W.l.o.g. assume that $0 \neq 1$. Then the projection $\pi_1(T)$ onto the first component is a subsemiring of $R$, where $0$ is an additive
identity and $1 \neq 0$ is a multiplicative identity.  
By these remarks,  the class of $\{0,1\}$-free finite semirings forms a pseudo-variety of finite semirings. 
Again, this fact will not be used in the rest of the paper, but it might be of independent interest.

\section{Circuit evaluation and main results} \label{sec-circuits}

We define circuits over general algebraic structures. Let $\mathcal{A} = (D,f_1, \dots, f_k)$ be an algebraic structure.
A {\em circuit} over  $\mathcal{A}$ is a triple $\C = (V,A_0,\rhs)$ where $V$ is a finite set of {\em gates},
	$A_0 \in V$ is the {\em output gate} and $\rhs$ (which stands for right-hand side) is a function that assigns to each gate $A \in V$ an element
	$a \in D$ or an expression of the form $f_i(A_1, \dots, A_n)$, where $n=n_i$ and $A_1, \dots, A_n \in V$ are called the {\em input gates for $A$}.
	Moreover, the binary relation $\{ (A,B) \in V \times V \mid A \text{ is an input gate for } B \}$ is required to be acyclic.
	Its reflexive and transitive closure is a partial order on $V$ that we denote with $\leq_{\C}$.
	Every gate $A$ evaluates to an element $[A]_\C \in A$ in the natural way:
	If $\rhs(A) = a \in D$, then $[A]_\C = a$ and if $\rhs(A) = f_i(A_1, \dots, A_n)$ then
	$[A]_\C = f_i([A_1]_\C, \dots, [A_n]_\C)$. Moreover, we define $[\C] = [A_0]_\C$ (the value computed by $\C$).
	If the circuit $\C$ is clear from the context, we also write $[A]$ instead of $[A]_\C$.
	We say that two circuits $\C_1$ and $\C_2$ over the structure $\mathcal{A}$ are {\em equivalent} if $[\C_1] = [\C_2]$.
	Sometimes we also use circuits without an output gate; such a circuit is just a pair $(V,\rhs)$.
	A subcircuit of $\C$ is the restriction of $\C$ to a downwards closed (w.r.t.~$\le_\C$) subset of $V$.
	A gate $A$ with $\rhs(A) = f_i(A_1, \dots, A_n)$ is called an {\em inner gate}, otherwise it is an {\em input gate} of $\C$.
	Quite often, we view a circuit as a directed acyclic graph, where the inner nodes are labelled with an operations $f_i$, and the leaf
	nodes are labelled with elements from $D$. In our proofs it is sometimes convenient to allow arbitrary terms built from $V \cup D$
	using the operations $f_1, \dots, f_k$ in right-hand sides. For instance, over a semiring $(R,+,\cdot)$ we might have
	$\rhs(A) = s\cdot B \cdot t  + C + s$ for $s,t \in R$ and $B,C \in V$.
	A circuit is in {\em normal form}, if all right-hand sides are of the form $a \in D$ or $f_i(A_1, \dots, A_n)$
	with $A_1, \ldots, A_n \in V$. We will make use of the following simple fact:

\begin{lemma} \label{lemma-normal-form}
A given circuit can be transformed in logspace into an equivalent normal form circuit.
\end{lemma}

\begin{proof}
The only non-trivial part is the elimination of copy gates $A$ with $\rhs(A) = B$ for a gate $B$; 
all other right-hand sides that violate the normal form have to be split up using fresh gates. This is easily
done in logspace. For copy gates consider the directed graph $G$ that contains for every copy gate $A$ the gate $A$ as well
as the gate $\rhs(A)$. Moreover, there is a directed edge from $A$ to $B=\rhs(A)$. This is a directed forest where the edges
are oriented towards the roots since every node has at most one outgoing edge (and the graph is acyclic).
By traversing all (deterministic) paths,
we can compute the reflexive transitive closure $G^*$ of $G$ in logspace. Using $G^*$ it is straightforward to eliminate copy gates:
For every copy gate $A$ we redefine $\rhs(A) = \rhs(B)$, where $B$ is the unique node in $G^*$ of outdegree zero such that 
$(A,B)$ is an edge of $G^*$.
\end{proof}
%Similarly to the above definition, one can define circuits over a semigroup. 
The {\em circuit evaluation problem} $\CEP(\mathcal{A})$ for some algebraic structure $\mathcal{A}$
(say a semigroup or a semiring) is the following computational problem:

\outputproblem{A circuit $\C$ over $\mathcal{A}$ and an element $a \in D$ from its domain.}{Decide whether $[\C] = a$.}

Note that for a finite structure $\mathcal{A}$, $\CEP(\mathcal{A})$ is basically equivalent to 
its computation variant, where one actually
computes the output value $[\C]$ of the circuit: if $\CEP(\mathcal{A})$ belongs to a complexity class
$\mathsf{C}$, then the computation variant belongs to $\AC^0(\mathsf{C})$, and if the latter belongs to
$\AC^0(\mathsf{C})$ then $\CEP(\mathcal{A})$ belongs to the decision fragment of $\AC^0(\mathsf{C})$.

Clearly, for every finite structure the circuit evaluation problem can be solved in polynomial time
by evaluating all gates along the partial order $\le_\C$.
Ladner's classical $\Ptime$-completeness result for the Boolean circuit value problem \cite{Lad75} can be stated as follows:

\begin{theorem}[\cite{Lad75}] \label{thm-ladner}
For the Boolean semiring $\mathbb{B}_2 = (\{0,1\}, \vee, \wedge)$, the problem $\CEP(\mathbb{B}_2)$ is 
$\Ptime$-complete.
\end{theorem}
For semigroups, the following dichotomy was shown in \cite{BeMcPeTh97}:

\begin{theorem}[\cite{BeMcPeTh97}]
	\label{thm:semigroup-classification}
	Let $S$ be a finite semigroup.
	\begin{itemize}
        \item If $S$ is aperiodic, then $\CEP(S)$ is in $\NL$.
         \item If $S$ is solvable, then $\CEP(S)$ belongs to $\DET$.
          \item If $S$ is not solvable, then $\CEP(S)$ is $\Ptime$-complete.
        \end{itemize}
\end{theorem}
Some remarks should be made:
\begin{itemize}
\item In \cite{BeMcPeTh97}, Theorem~\ref{thm:semigroup-classification} is only shown for monoids, but the extension to semigroups is straightforward:
If the finite semigroup $S$ has a non-solvable subgroup, then $\CEP(S)$ is $\Ptime$-complete, since the circuit evaluation problem for a non-solvable finite
group is $\Ptime$-complete. On the other hand, if $S$ is solvable (resp., aperiodic), then also the monoid $S^1$ is solvable  (resp., aperiodic).
This holds, since the subgroups of $S^1$ are exactly
the subgroups of $S$ together with $\{1\}$. Hence, $\CEP(S^1)$ is in $\DET$ (resp., $\NL$), which implies that $\CEP(S)$ is in $\DET$ (resp., $\NL$).
\item In \cite{BeMcPeTh97}, the authors use the original definition $\DET = \NC^1(\mathsf{det})$ of Cook. But the arguments in \cite{BeMcPeTh97}
actually show that for a finite solvable semigroup, $\CEP(S)$ belongs to $\AC^0(\mathsf{det})$ (which is our definition of $\DET$).
\item In \cite{BeMcPeTh97}, the authors study two versions of the circuit evaluation problem for a semigroup $S$: What we call $\CEP(S)$ is called
$\UCEP(S)$ (for ``unrestricted circuit evaluation problem'') in \cite{BeMcPeTh97}. The problem $\CEP(S)$ is defined in  \cite{BeMcPeTh97} as 
the circuit evaluation problem, where in addition the input circuit must have the property that the output gate has no ingoing edges and
all gates are reachable from the output gate. These conditions can be enforced with an $\AC^0(\NL)$-precomputation. Hence, the difference between the two 
variants is only relevant for classes below $\NL$. We only consider the unrestricted version of the circuit evaluation problem (where the input circuit 
is arbitrary). To keep notation simple, we decided to refer with $\CEP$ to the unrestricted version.
\end{itemize}
Let us fix a {\em finite} semiring $R = (R,+,\cdot)$ for the rest of the paper.
Note that $\CEP(\aR)$ (resp., $\CEP(\mR)$) is the restriction of $\CEP(R)$ to circuits without multiplication (resp.,
addition) gates.
Since every commutative semigroup is solvable, Theorem~\ref{thm:semigroup-classification} implies that
$\CEP(\aR)$ belongs to $\DET$.
The main result of this paper is:

\begin{theorem}
	\label{thm-main}
	If the finite semiring $R$ is $\{0,1\}$-free, then the problem $\CEP(R)$ belongs to the class $\AC^0(\NL, \CEP(\aR), \CEP(\mR))$.
	Otherwise $\CEP(R)$ is $\Ptime$-complete.
\end{theorem}
Note that if $\CEP(\aR)$ or $\CEP(\mR)$ is $\NL$-hard, then 
$$\AC^0(\NL, \CEP(\aR), \CEP(\mR)) = \AC^0(\CEP(\aR), \CEP(\mR)).$$
For example, this is the case, if $\aR$ or $\mR$ is an aperiodic nontrivial monoid \cite[Proposition~4.14]{BeMcPeTh97}.

The $\Ptime$-hardness statement in Theorem~\ref{thm-main} is easy to show:
\begin{lemma} \label{lemma-P-hard}
	If the finite semiring $R$ is not $\{0,1\}$-free, then the problem $\CEP(R)$ is $\Ptime$-complete.
\end{lemma}

\begin{proof}
	By Lemma~\ref{lemma-0-1-free}, $R$ contains either $\B_2$ or $\Z_d$ for some $d \ge 2$.
	In the former case, $\Ptime$-hardness follows from Ladner's theorem.
	Furthermore, one can reduce the $\Ptime$-complete Boolean circuit value problem
		over $\{0,1,\wedge,\neg\}$ to $\CEP(\Z_d)$ for $d \geq 2$:
	A gate $z = x \wedge y$ is replaced by $z = x \cdot y$ and a gate $y = \neg x$ is replaced by $y = 1 + (d-1) \cdot x$.
\end{proof}
Theorem~\ref{thm:semigroup-classification} and \ref{thm-main} yield the following corollaries:
\begin{corollary}
	\label{coro-main}
	Let $R$ be a finite semiring. 
	\begin{itemize}
         \item If $R$ is not $\{0,1\}$-free or $\mR$ is not solvable, then $\CEP(R)$ is $\Ptime$-complete.
         \item If $R$ is $\{0,1\}$-free and $\mR$ is solvable, then $\CEP(R)$ belongs to $\DET$.
	\item If $R$ is $\{0,1\}$-free and $\mR$ is aperiodic,  then $\CEP(R)$ belongs to $\NL$.
        \end{itemize}
\end{corollary}

Let us present an application of Corollary~\ref{coro-main}. 
\begin{example} \label{ex-power-semigroup} 
An important semigroup construction found in the literature is the  power construction.
For a finite semigroup $S$ one defines the {\em power semiring}
$\mathcal{P}(S) = (2^S \setminus \{ \emptyset \}, \cup, \cdot)$ with the multiplication $A \cdot B = \{ ab \mid a \in A, b \in B \}$.
Notice that if one includes the empty set, then the semiring would not be $\{0,1\}$-free:
Take an idempotent $e \in S$. Then $\emptyset$ and $\{e\}$ form a copy of $\mathbb{B}_2$. Hence, the circuit evaluation problem is 
$\Ptime$-complete.

Let us further assume that $S$ is a monoid with identity $1$ (the general case will be considered below).
If $S$ contains an idempotent $e \neq 1$ then also $\mathcal{P}(S)$ is not $\{0,1\}$-free: $\{e\}$ and $\{1,e\}$ form a copy of $\mathbb{B}_2$.
On the other hand, if $1$ is the unique idempotent of $S$, then $S$ must be a group $G$.  Assume that $G$ is solvable; otherwise $\mathcal{P}(G)_{{\scriptscriptstyle \bullet}}$
is not solvable as well and has a $\Ptime$-complete circuit evaluation problem by Theorem~\ref{thm:semigroup-classification}. It is not hard to show that 
the subgroups of  $\mathcal{P}(G)_{{\scriptscriptstyle \bullet}}$ correspond to the quotient groups of subgroups of $G$; see also \cite{MCHa73}. Since $G$ is solvable
and the class of solvable groups is closed under taking subgroups and quotients, 
$\mathcal{P}(G)_{{\scriptscriptstyle \bullet}}$ is a solvable monoid. Moreover $\mathcal{P}(G)$ is $\{0,1\}$-free: 
Otherwise, Lemma~\ref{lemma-0-1-free} implies that there
are non-empty subsets $A, B \subseteq G$ such that $A \neq B$, $A \cup B = B$ (and thus $A \subsetneq B$), $A B = B A = A^2 = A$, and $B^2 = B$.
Hence, $B$ is a subgroup of $G$ and $A \subseteq B$. But then $B = A B = A$, which is a contradiction. By Corollary~\ref{coro-main},
$\CEP(\mathcal{P}(G))$ for a finite solvable group $G$ belongs to $\DET$. 
\end{example}
Let us now classify the complexity of $\CEP(\mathcal{P}(S))$ for arbitrary semigroups $S$.
A semigroup $S$ is called a {\em local group} if for all $e \in E(S)$ the local monoid $eSe$ is a group.
It is known that in every finite local group $S$ of size $n$ the minimal semigroup ideal is $S^n = S E(S) S$, see \cite[Proposition 2.3]{almeida09}.

\begin{theorem}
	\label{thm:power-semiring}
	Let $S$ be a finite semigroup.
	If $S$ is a local group and solvable, then $\CEP(\mathcal{P}(S))$ belongs to $\DET$.
	Otherwise $\CEP(\mathcal{P}(S))$ is $\Ptime$-complete.
\end{theorem}

\begin{proof}
	Let $S$ be a finite local group which is solvable. By \cite[Corollary 2.7]{Auinger20051}
	the multiplicative semigroup $\mathcal{P}(S)_{{\scriptscriptstyle \bullet}}$ is solvable as well.
	It remains to show that the semiring $\mathcal{P}(S)$ is $\{0,1\}$-free: 
	Towards a contradiction assume that $\mathcal{P}(S)$ is not $\{0,1\}$-free.
	By Lemma~\ref{lemma-0-1-free}, there exist non-empty sets
	$A \subsetneq B \subseteq S$ such that $A B = B A = A^2 = A$ and $B^2 = B$.
	Hence, $B$ is a subsemigroup of $S$, which is also a local group, and $A$ is a semigroup ideal in $B$.
	Since the minimal semigroup ideal of $B$ is $B^n$ for $n = |B|$ and $B^n = B$, we obtain $A=B$, which is a contradiction.
	
	The case that $S$ is not a local group follows from the arguments in Example~\ref{ex-power-semigroup}. In that case,
	there exists a local monoid $eSe$ which is not a group and hence contains an idempotent $f \neq e$.
	Since $\{ \{f\}, \{e,f\} \}$ forms a copy of $\B_2$ it follows that $\CEP(\mathcal{P}(S))$ is $\Ptime$-complete. Finally, 
	if $S$ is not solvable, then also $\mathcal{P}(S)^{\scriptscriptstyle}$ is not solvable and
	$\CEP(\mathcal{P}(S))$ is $\Ptime$-complete by Theorem~\ref{thm:semigroup-classification}.
\end{proof}

	\section{Proof of Theorem~\ref{thm-main}}
	
	The proof of Theorem~\ref{thm-main} will proceed in two steps. In the first step we reduce the problem to evaluating
	circuits in which the computation respects a type-function defined in the following.
	In the second step, we show how to evaluate such circuits.

	\begin{definition} \label{def-type}
		Let $E = E(R)$ be the set of multiplicative idempotents.
		Let $\C=(V,\rhs)$ be a circuit in normal form such that $[A]_{\C} \in ERE$ for all gates $A \in V$.
		A type-function for $\C$ is a mapping $\mathsf{type} : V \to E \times E$ such that the following conditions hold:
		\begin{itemize}
			\item For every $A \in V$ such that  $\mathsf{type}(A) = (e,f)$
			we have $[A]_{\C} \in eRf$.
			\item For every addition gate $A \in V$ with $\rhs(A) = B+C$ we have $\mathsf{type}(A) = \mathsf{type}(B) = \mathsf{type}(C)$.
			\item For every multiplication gate $A \in V$ with $\rhs(A) = B \cdot C$, $\mathsf{type}(B) = (e,e')$,  and $\mathsf{type}(C) = (f',f)$
			we have $\mathsf{type}(A) = (e,f)$.
		\end{itemize}
		A circuit is called {\em type admitting} if it admits a type-function.
	\end{definition}
	Note that we do not need an output gate in a type admitting circuit.
	
	\begin{definition}
		A function $\alpha: R^m \to R$ ($m \geq 0$) is called {\em affine} if there are $a_1, b_1, \ldots, a_m,b_m,c \in R$
		such that for all $x_1, \ldots, x_m \in R$:
		\[
			\alpha(x_1, \dots, x_m) = \sum_{i=1}^m a_i x_i b_i + c \quad \text{ or } \quad \alpha(x_1, \dots, x_m) = \sum_{i=1}^m a_i x_i b_i .
		\]
	\end{definition}
	We represent this affine function by the tuple $(a_1,b_1, \ldots, a_m,b_m,c)$ or $(a_1,b_1, \ldots, a_m,b_m)$.
	
	Theorem~\ref{thm-main} is now an immediate corollary of the following two propositions (and the obvious fact
	that an affine function with a constant number of inputs can be evaluated in $\AC^0$).
	
	\begin{proposition}
		\label{thm:step-1}
		Given a circuit $\C$ over the finite semiring $R$, one can compute the following data in $\AC^0(\NL,\CEP(\aR))$:
		\begin{itemize}
			\item an affine function $\alpha: R^m \to R$ for some $0 \leq m \leq |R|^4$,
			\item a type admitting circuit $\C' = (V',\rhs')$, and
			\item a list of gates $A_1, \dots, A_m \in V'$ such that $[\C] = \alpha([A_1]_{\C'}, \dots, [A_m]_{\C'})$.
		\end{itemize}
	\end{proposition}
	
	\begin{proposition}
		\label{thm:step-2}
		If $R$ is $\{0,1\}$-free, then the restriction of $\CEP(R)$ to type admitting circuits is in $\AC^0(\NL, \CEP(\aR), \CEP(\mR))$.
	\end{proposition}
	
	 It is not clear how to test efficiently whether a circuit is type admitting. But this is not a problem for us, 
	since we will apply Proposition~\ref{thm:step-2} only to circuits resulting from Proposition~\ref{thm:step-1},
	which are type admitting by construction.

\subsection{Step 1: Reduction to typing admitting circuits}

In this section, we prove Proposition~\ref{thm:step-1}.
Let $\C$ be a circuit in normal form over our fixed finite semiring $R=(R,+,\cdot)$ of size $n = |R|$.
We assume that $n \geq 2$ (the case $n=1$ is trivial).
Throughout the section we will use $E = E(R)$.
Note that $R^n = RER$ is closed under multiplication with elements from $R$.
Thus, $\langle R^n \rangle$ is an ideal. Every element $a \in \langle R^n \rangle$ 
can be written as a finite sum $a = \sum_{i=1}^k a_i$ with $a_i \in R^n$. Moreover,
since $R$ is a fixed finite semiring, the number $k$ of summands can be bounded 
by a constant that only depends on $R$.

The reduction to type admitting circuits is done in two steps:
$$
\text{circuit over $R$} \ \xrightarrow{\text{Lemma~\ref{lemma:monomlength}}} \ \text{circuit over $\langle R^n \rangle = \langle RER \rangle$} \ \xrightarrow{\text{Lemma~\ref{lemma:FSF}}}   \
\text{type admitting circuit}
$$
Before we prove Lemma~\ref{lemma:monomlength} and \ref{lemma:FSF}, we prove a lemma that allows to eliminate certain input values from a circuit:

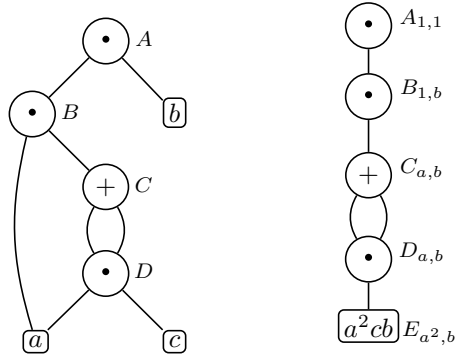
\begin{figure}[t]
	
	\centering
	\ifmac
	\begin{tikzpicture}
	\node[gate, label={[above right = -.5cm and .5cm of A]\footnotesize $A$}] (A) {\Huge $\cdot$};
	\node[gate, label={[above right = -.5cm and .5cm of B]\footnotesize $B$}, below left = 15pt and 15pt of A] (B) {\Huge $\cdot$};
	\node[gate, label={[above right = -.5cm and .5cm of C]\footnotesize $C$}, below right = 15pt and 15pt of B] (C) {$+$};
	\node[gate, label={[above right = -.5cm and .5cm of D]\footnotesize $D$}, below = 15pt of C] (D) {\Huge $\cdot$};
	\node[input, below left = 15pt and 15pt of D] (E) {$a$};
	\node[input, below right = 15pt and 15pt of D] (F) {$c$};
	\node[input, below right = 15pt and 15pt of A] (G) {$b$};
	
	\draw (A) -- (B);
	\draw (A) -- (G);
	\draw (B) -- (C);
	\draw [bend right = 15] (B) to (E);
	\draw [bend left] (C) to (D);
	\draw [bend right] (C) to (D);
	\draw (D) -- (E);
	\draw (D) -- (F);
	
	\end{tikzpicture}
	\hspace{5em}
	\begin{tikzpicture}
	\node[gate, label={[above right = -.5cm and .7cm of A]\footnotesize $A_{1,1}$}] (A) {\Huge $\cdot$};
	\node[gate, label={[above right = -.5cm and .7cm of B]\footnotesize $B_{1,b}$}, below = 9pt of A] (B) {\Huge $\cdot$};
	\node[gate, label={[above right = -.5cm and .7cm of C]\footnotesize $C_{a,b}$}, below = 12pt of B] (C) {$+$};
	\node[gate, label={[above right = -.5cm and .7cm of D]\footnotesize $D_{a,b}$}, below = 14pt of C] (D) {\Huge $\cdot$};
	
	\node[input, label={[above right = -.6cm and .8cm of E]\footnotesize $E_{a^2,b}$}, below = 10pt of D] (E) {$a^2cb$};

	\draw (A) -- (B);
	\draw (B) -- (C);
	\draw [bend left] (C) to (D);
	\draw [bend right] (C) to (D);
	\draw (D) -- (E);
	
	\end{tikzpicture}
	\else
	\begin{tikzpicture}
	\node[gate, label={[above right = -.5cm and .3cm of A]\footnotesize $A$}] (A) {\Huge $\cdot$};
	\node[gate, label={[above right = -.5cm and .3cm of B]\footnotesize $B$}, below left = 15pt and 15pt of A] (B) {\Huge $\cdot$};
	\node[gate, label={[above right = -.5cm and .3cm of C]\footnotesize $C$}, below right = 15pt and 15pt of B] (C) {$+$};
	\node[gate, label={[above right = -.5cm and .3cm of D]\footnotesize $D$}, below = 15pt of C] (D) {\Huge $\cdot$};
	\node[input, below left = 15pt and 15pt of D] (E) {$a$};
	\node[input, below right = 15pt and 15pt of D] (F) {$c$};
	\node[input, below right = 15pt and 15pt of A] (G) {$b$};
	
	\draw (A) -- (B);
	\draw (A) -- (G);
	\draw (B) -- (C);
	\draw [bend right = 15] (B) to (E);
	\draw [bend left] (C) to (D);
	\draw [bend right] (C) to (D);
	\draw (D) -- (E);
	\draw (D) -- (F);
	
	\end{tikzpicture}
	\hspace{5em}
	\begin{tikzpicture}
	\node[gate, label={[above right = -.5cm and .4cm of A]\footnotesize $A_{1,1}$}] (A) {\Huge $\cdot$};
	\node[gate, label={[above right = -.5cm and .4cm of B]\footnotesize $B_{1,b}$}, below = 9pt of A] (B) {\Huge $\cdot$};
	\node[gate, label={[above right = -.5cm and .4cm of C]\footnotesize $C_{a,b}$}, below = 12pt of B] (C) {$+$};
	\node[gate, label={[above right = -.5cm and .4cm of D]\footnotesize $D_{a,b}$}, below = 14pt of C] (D) {\Huge $\cdot$};
	
	\node[input, label={[above right = -.6cm and .5cm of E]\footnotesize $E_{a^2,b}$}, below = 10pt of D] (E) {$a^2cb$};

	\draw (A) -- (B);
	\draw (B) -- (C);
	\draw [bend left] (C) to (D);
	\draw [bend right] (C) to (D);
	\draw (D) -- (E);
	
	\end{tikzpicture}
	\fi
	\caption{Illustration of the proof of Lemma~\ref{lemma:leaf-removement} with the ideal $I = \{c\}$ in $S = \{a,b,c\}$.}
	\label{fig-removal-M}
\end{figure}

\begin{lemma}
	\label{lemma:leaf-removement}
	Assume that $I \subseteq R$ is a non-empty ideal of $R$.
	Let $\C= (V,A_0,\rhs)$ be a circuit in normal form. 
	Consider the set $U = \{ A \in V \mid A \text{ is an inner gate or } \rhs(A) \in I \}$ and 
	assume that $A_0 \in U$ and for all $A,B,C\in V$ the following holds:
	\begin{itemize}
		\item If $\rhs(A)=B \cdot C$ then $B \in U$ or $C \in U$.
		\item If $\rhs(A)=B + C$ then $B,C \in U$.
	\end{itemize}
	There is a logspace-computable function that returns for 
	a given circuit $\C$ with the above properties an equivalent circuit $\D$ in normal form over the ideal (and hence subsemiring) $I$.
\end{lemma}

\begin{proof}
	Let $U \subseteq V$ be defined as in the lemma.
	We first compute in logspace a circuit $\C' = (V',(A_0)_{1,1},\rhs')$ which contains a gate $A_{\ell,r}$
	for all gates $A \in U$ and $\ell,r \in R^1$ (recall that $R^1$ is $R$ together with a fresh multiplicative identity $1$)
	such that $[A_{\ell,r}]_{\C'} = \ell \cdot [A]_\C \cdot r$.
	Note that $V'$ is non-empty since $A_0 \in U$.
	The gates of $\C'$ are indexed by elements of $R^1$ instead of $R$ to simplify the notation in the following.
	To define the right-hand sides let us take a gate $A \in U$ and $\ell,r \in R^1$.
	
	\medskip
	
	\noindent
	{\em Case 1.} $\rhs(A) = a \in I$. We set $\rhs'(A_{\ell,r}) = \ell a r$. Note that
	$\ell a r \in I$, since $I$ is an ideal in $R$.
	
	\medskip
	
	\noindent
	{\em Case 2.} $\rhs(A) = B + C$. Then we must have $B,C \in U$ and we set $\rhs'(A_{\ell,r}) = B_{\ell,r} + C_{\ell,r}$.
	
	\medskip
	
	\noindent
	{\em Case 3.}  $\rhs(A) = B \cdot C$. Then $B \in U$ or $C \in U$.
	If both $B$ and $C$ belong to $U$, then we set $\rhs'(A_{\ell,r}) = B_{\ell,1} \cdot C_{1,r}$.
	If $C \not\in U$ then $B \in U$ and $\rhs(C) = c \in R \setminus I$. We set $\rhs'(A_{\ell,r}) = B_{\ell,cr}$. 
	The case that $B\not\in U$ is symmetric.
	
	\medskip
	
	\noindent
	By Lemma~\ref{lemma-normal-form} we can transform 
	the circuit $\C'$ in logspace into a normal form circuit $\D$.
	The correctness of the above construction can be easily shown by induction along the partial order $\leq_{\C}$.
\end{proof}
Figure~\ref{fig-removal-M} illustrates the construction for $R = \{a,b,c\}$ and the ideal $I = \{c\}$ (the concrete semiring structure
is not important). Note that $a^2 c b \in I$. The circuit on the right-hand side is only the part of the constructed circuit that is rooted in the 
gate $A_{1,1}$. Copy gates are not eliminated.

A circuit $\mathcal{A} = (V,A_0,\rhs)$ over the semiring $(\mathbb{N},+,\cdot)$ is also called an {\em arithmetic circuit}. 
We further assume that all input gates of an arithmetic circuit have a right-hand side from $\{0,1\}$. W.l.o.g.~we can also
assume that there are unique input gates with right-hand sides $0$ and $1$, respectively.
The {\em multiplication depth} of a
gate $A \in V$ is the maximal number of multiplication gates along every path from an input gate to $A$ (where $\mathcal{A}$
is viewed as a directed acyclic graph). The multiplication depth of $\mathcal{A}$ is the maximal multiplication depth among all gates in $\mathcal{A}$.   
An {\em addition circuit} is an arithmetic circuit without multiplication gates.
The following lemma is shown in \cite[Lemma~6]{KonigL15}:

\begin{lemma} \label{lemma-cocoon}
	Let $c$ be a fixed constant.
	A given arithmetic circuit $\mathcal{A}$ of multiplication depth at most $c$, where in addition every multiplication gate is labelled with its multiplication depth,
	can be transformed in logspace into an addition circuit $\mathcal{A}'$ that contains every gate $A$ of $\mathcal{A}$ and moreover satisfies $[A]_{\mathcal{A}} = [A]_{\mathcal{A}'}$.
\end{lemma}

For the following proofs, it is sometimes more convenient to evaluate a circuit $\C=(V,A_0,\rhs)$ over the free semiring $\N[R]$
generated by the set $R$.\footnote{Of course, there is no chance of efficiently evaluating a circuit over the free semiring $\N[R]$, since this might produce
	a doubly exponential number of monomials of exponential length. Circuit evaluation over $\N[R]$ is only used as a tool in our proofs.}
Recall that this semiring consists of all mappings $f : R^+ \to \N$ with finite and non-empty support, where $R^+$ consists of all finite 
non-empty words over the alphabet $R$. 
So, there are two ways to evaluate $\C$: We can evaluate $\C$ over $R$ (and this is our main interest) and
we can evaluate $\C$ over $\N[R]$. In order to distinguish these two ways of evaluation, we write 
$\fval{A}_\C \in \N[R]$ for the value of gate $A\in V$ in $\C$, when $\C$ is evaluated in $\N[R]$.
Again we omit the index $\C$ if it is clear from the context. Moreover, $\fval{\C} = \fval{A_0}_{\C}$.
Note that $\fval{A}_\C$ is a mapping from $R^+$ to $\N$; hence for a word $w \in R^+$, $\fval{A}_\C(w)$ is a natural number.
Let $h$ be the canonical semiring homomorphism from $\N[R]$ to $R$ that evaluates a non-commutative polynomial
in the semiring $R$ (the range of this mapping is the subsemiring
generated by $R$). Thus, we have $[A]_\C = h(\fval{A}_\C)$ for every gate $A$ and $[\C] = h(\fval{\C})$.
An example of a free evaluation of a circuit is shown in Figure~\ref{fig-short-long} on the left.

Recall that $|R| = n \geq 2$.
We define 
$$R^{<n} = \{ w \in R^+ \mid |w| < n\} \ \text{ and } \ R^{\geq n} = \{ w \in R^+ \mid |w| \geq n\}.$$
Note that
these are sets of finite words over the alphabet $R$.  So, these notations should not be confused with the notation
$R^n$, which is a subset of $R$ (the set of all $n$-fold products). In fact, we have $h(R^{\ge n}) = R^n = RER$.
For every non-commutative polynomial $f \in \N[R]$ we define $f^\sigma, f^\lambda \in \N[R] \cup \{\bot\}$
(the short part and the long part of $f$) as follows ($\bot$ is a new symbol that stands for ``undefined''):
\begin{enumerate}
	\item If $\supp(f) \subseteq R^{< n}$, then $f^\sigma = f$ and $f^\lambda = \bot$.
	\item If $\supp(f) \subseteq R^{\geq n}$, then $f^\sigma = \bot$ and $f^\lambda = f$.
	\item \label{decomp}
	Otherwise let $f^\sigma,f^\lambda \in \N[R]$ such that $f = f^\sigma + f^\lambda$,
	$\supp(f^\sigma) \subseteq R^{< n}$ and $\supp(f^\lambda) \subseteq R^{\ge n}$.
\end{enumerate}
Note that either $f^\sigma \neq \bot$ or $f^\lambda \neq \bot$, and that the decomposition in \ref{decomp} is unique.
Moreover, if $f^\sigma \neq\bot\neq f^\lambda$, then $f = f^\sigma + f^\lambda$.

\begin{example}
	Let $R = \{a,b,c\}$ and thus $n = 3$. Let $f =  2 abbca + 3 caab + bab + 4 ac + 7b \in \N[\{a,b,c\}]$.
	We have $f^\sigma = 4 ac + 7b$
	and $f^\lambda = 2 abbca + 3 caab + bab$.
\end{example}

\begin{lemma}
	\label{lemma:monomlength}
	There is a function in $\AC^0(\NL,\CEP(\aR))$ that returns for a given  circuit
	$\C= (V,A_0,\rhs)$ either
	\begin{itemize}
		\item the semiring element $[\C] \in R$ (namely if $\fval{\C}^\lambda = \bot)$, or
		\item a circuit $\D$ over the subsemiring $\langle R^n \rangle = \langle RER \rangle$ such that $[\C] = [\D]$
		(namely if $\fval{\C}^\sigma = \bot)$, or
		\item a circuit $\D$ over the subsemiring $\langle R^n \rangle = \langle RER \rangle$ and a semiring element $\sigma \in R$ such that $[\C] = [\D] + \sigma$
		(namely if $\fval{\C}^\sigma \neq \bot \neq \fval{\C}^\lambda)$.
	\end{itemize}
\end{lemma}

\begin{proof}
   By Lemma~\ref{lemma-normal-form} we can assume that $\C$ is in normal form.
	In the following, we omit the index $\C$ in $[A]_\C$ and $\fval{A}_\C$, where $A \in V$ is a gate of the circuit $\C$.
	
	\medskip
	\noindent
	{\em Step 1.}
	We first compute in $\AC^0(\NL)$ the set of all gates $A \in V$ such that $\fval{A}^\sigma \neq \bot$.
	For this, we construct in logspace an arithmetic circuit $\mathcal{A}$ 
	over the semiring $(\N,+,\cdot)$ with gates $A_w$ where $A \in V$ and $w \in R^{< n}$
	such that $[A_w]_{\mathcal{A}} = \fval{A}(w)$ as follows:
	\begin{itemize}
		\item If $\rhs(A) = a \in R$ then
		$\rhs(A_w) = \begin{cases} 1 & \text{if }w=a \\ 0 & \text{otherwise.}\end{cases}$
		\item If $\rhs(A) = B+C$ then $\rhs(A_w) = B_w + C_w$.
		\item If $\rhs(A) = B \cdot C$ then $\rhs(A_w) = \sum_{w=uv} B_u \cdot C_v$, where the sum goes over all $u,v \in R^{< n}$ with $w = uv$.
	\end{itemize}
	Note that the empty sum is interpreted as $0$ and that $\mathcal{A}$ has constant multiplication depth.
	Moreover, the multiplication depth of a gate $A_w$ is $|w|-1$.
	By Lemma~\ref{lemma-cocoon}  we can transform in logspace $\mathcal{A}$ into an equivalent addition circuit, which 
	we still denote with $\mathcal{A}$.
	The circuit $\mathcal{A}$ contains all gates $A_w$ ($A \in V$, $w \in R^{< n}$) and possibly some additional gates.
	
	We can assume that $\mathcal{A}$ has a unique input gate $Z$ with right-hand side $1$.
	Let $U_\sigma$ be the set of all gates $X$ of $\mathcal{A}$ such that
	$Z \le_{\mathcal{A}} X$. These are exactly those gates of $\mathcal{A}$
	that evaluate to a number larger than zero. Hence, for all $A \in V$ and $w \in R^{< n}$, we have
	$A_w \in U_\sigma$ if and only if $\fval{A}(w) > 0$.
	Moreover, $\fval{A}^\sigma \neq \bot$ if and only if $A_w \in U_\sigma$ for some $w \in R^{< n}$.
	The set $U_\sigma$ can be computed in $\AC^0(\NL)$. Hence, we can also compute for every $A \in V$ the information
	whether  $\fval{A}^\sigma \neq \bot$ and, in case $\fval{A}^\sigma \neq \bot$, the set $\supp(\fval{A}^\sigma) =  \supp(\fval{A}) \cap R^{<n}$.
	
	\medskip
	\noindent
	{\em Step 2.}
	For each gate $A \in V$ with $\fval{A}^\sigma \neq \bot$ we now compute 
	the semiring element $h(\fval{A}^\sigma) \in R$. For this we construct in logspace a circuit over $\aR$
	that evaluates to $h(\fval{A}^\sigma)$.  Hence,  $h(\fval{A}^\sigma)$ can be computed using oracle access to $\CEP(\aR)$.
	
	We first remove from the arithmetic addition circuit $\mathcal{A}$ all gates that are not in $U_\sigma$. This may
	produce copy gates (which is not a problem). Moreover, gate $Z$ is now the only input gate of $\mathcal{A}$.
	For a semiring element $a \in R$ we define the circuit $\mathcal{C}_a$ (over $\aR$) by taking the addition circuit  $\mathcal{A}$
	and redefining $\rhs_\sigma(Z) = a$. Then, for every gate $A_w \in U_\sigma$ ($A \in V$, $w \in R^{< n}$) we have
	$[A_w]_{\mathcal{C}_a} =  \fval{A}(w) \cdot a$.
	In particular, if $\fval{A}^\sigma \neq \bot$, then
	$$
	h(\fval{A}^\sigma) = \sum_{w \in \supp(\fval{A}) \cap R^{<n}} \!\!\!\!\!\!\!\! \fval{A}(w) \cdot h(w) \ = 
	\sum_{w \in \supp(\fval{A}) \cap R^{<n}} \!\!\!\!\!\!\!\! [A_w]_{\mathcal{C}_{h(w)}} .
	$$
	From the circuits $\mathcal{C}_{h(w)}$ we can construct in logspace a circuit over $\aR$ for this semiring element. 	
	Evaluating this circuit using oracle access to $\CEP(\aR)$ yields the element $h(\fval{A}^\sigma)$.
	
	\medskip
	\noindent
	{\em Step 3.}
	Next, we compute in $\AC^0(\NL)$ the set of all gates $A \in V$ such that $\fval{A}^\lambda \neq \bot$.
	Since $n \geq 2$ (our initial assumption on the semiring $R$), we have $\fval{A}^\lambda \neq \bot$
	if and only if there exist a gate $A' \leq_{\C} A$ with
	$\rhs(A') = B \cdot C$ and words $w_1, w_2 \in R^{< n}$ such that $|w_1 w_2| \geq n$,
	$\fval{B}(w_1) > 0$ (i.e., $B_{w_1} \in U_\sigma$) and $\fval{C}(w_2) > 0$ (i.e., $C_{w_2} \in U_\sigma$).
	This condition can be tested in $\NL$. 
	Hence, we can assume that the set of all $A \in V$ with
	$\fval{A}^\lambda \neq \bot$ is computed. If $\fval{A_0}^\lambda = \bot$, then we must have $\fval{A_0}^\sigma \neq \bot$
	and we return the previously computed semiring element 	$h(\fval{A_0}^\sigma)$, which is $[\C]$ in this case.
	Let us now assume that $\fval{A_0}^\lambda \neq \bot$.

	\medskip
	\noindent
	{\em Step 4.}
	We then construct a circuit $\C_\lambda = (V_\lambda, (A_0)_{\lambda}, \rhs_\lambda)$, which contains 
	for every gate $A \in V$ with $\fval{A}^\lambda \neq \bot$ a gate $A_\lambda$ such that 
	$[A_\lambda]_{\C_\lambda} = h(\fval{A}^\lambda)$.  In particular, $[\C_\lambda] = 
	h(\fval{\C}^\lambda)$.
	
	In a first step, we compute in $\AC_0$ the set $M^\lambda$ of all multiplication gates $A \in V$ such that
	the following conditions hold:  $\fval{A}^\lambda \neq \bot$,
	$\rhs(A) = B \cdot C$ for $B,C \in V$, $\fval{B}^\sigma \neq \bot \neq \fval{C}^\sigma$, and 
	there exist $u  \in \supp(\fval{B}^\sigma)$, $v \in \supp(\fval{C}^\sigma)$ with $|uv| \ge n$.
	This means that in the product $\fval{B} \cdot \fval{C}$ a monomial of length at least $n$ 
	arises from monomials $u \in \supp(\fval{B})$, $v \in \supp(\fval{C})$, both of which have length smaller than $n$.
	
	Next, for every multiplication gate $A \in M^\lambda$, where $\rhs(A) = B \cdot C$ we compute in $\AC_0(\CEP(\aR))$ the 
	semiring element
	$$
	m_A := \sum_{u,v} (\fval{B}(u) \fval{C}(v)) \cdot h(uv) \in \langle R^n \rangle,
	$$
	where the sum is taken over all words $u \in \supp(\fval{B}^\sigma)$, $v \in \supp(\fval{C}^\sigma)$ with $|uv| \ge n$.
	Let us take the addition circuit $\mathcal{A}$ constructed in Step 1 and 2 above. 
	We add to $\mathcal{A}$  a single layer of multiplication gates $A_{u,v}$, where
	$B_u, C_v \in U_\sigma$. The right-hand side of $A_{u,v}$ is $B_u \cdot C_v$. 
	Using Lemma~\ref{lemma-cocoon}, we can transform this circuit in logspace into an equivalent
	addition circuit; let us denote this circuit with $\mathcal{A}^2$. Clearly, gate $A_{u,v}$ evaluates
	to the number $\fval{B}(u)\fval{C}(v) \in \N$. This number is larger than zero, since $B_u, C_v \in U_\sigma$.
	Hence, by replacing  in the addition circuit $\mathcal{A}^2$ the input value $1$ by the semigroup element $h(uv)$,
	we obtain  a circuit over the semigroup $\aR$, which we can evaluate using oracle access to $\CEP(\aR)$.
	The value of gate $A_{u,v}$ yields the semiring element $(\fval{B}(u) \fval{C}(v)) \cdot h(uv)$.
	Finally, the sum of all these values (for all $u \in \supp(\fval{B}^\sigma)$, $v \in \supp(\fval{C}^\sigma)$ with $|uv| \ge n$)
	can be computed by another $\AC^0$-computation (it is a sum of a constant number of semiring elements).
	
	It remains to define the  right-hand sides of the gates $A_\lambda$ in $\C_\lambda$.
	We distinguish the following cases (note that $A$ must be an inner gate of $\C$ since $n \geq 2$ and $\C$ is in normal form).
	
	\medskip
	\noindent
	{\em Case 1.} $\rhs(A) = B+C$. Then we must have $\fval{B}^\lambda \neq \bot$ or 
	$\fval{C}^\lambda \neq \bot$ (otherwise $\fval{A}^\lambda = \bot$)
	and we set
	\[
	\rhs_\lambda(A_\lambda) =
	\begin{cases}
	B_\lambda, & \text{if } \fval{C}^\lambda = \bot, \\
	C_\lambda, & \text{if } \fval{B}^\lambda = \bot, \\
	B_\lambda + C_\lambda, & \text{otherwise}.
	\end{cases}
	\]
	{\em Case 2.} $\rhs(A) = B \cdot C$, $A \in M^\lambda$, and $\bot \notin \{\fval{B}^\sigma, \fval{B}^\lambda,\fval{C}^\sigma, \fval{C}^\lambda\}$. Then we set
	\begin{equation}
	\label{eq:rhs-mult}
	\rhs_\lambda(A_\lambda) = B_\lambda \cdot C_\lambda + h(\fval{B}^\sigma) \cdot C_\lambda + B_\lambda \cdot h(\fval{C}^\sigma) + m_A .
	\end{equation}
	If $A \not\in M^\lambda$ but $\bot \notin \{\fval{B}^\sigma, \fval{B}^\lambda,\fval{C}^\sigma, \fval{C}^\lambda\}$ then we take the same definition
	but omit the summand $m_A$.
		
	Let us explain the definition \eqref{eq:rhs-mult}. We have 
	\begin{eqnarray*}
		\fval{A}^\sigma + \fval{A}^\lambda &= &\fval{A} \\
		& = & \fval{B} \cdot \fval{C} = (\fval{B}^\sigma + \fval{B}^\lambda) \cdot (\fval{C}^\sigma + \fval{C}^\lambda) \\
		& = & \fval{B}^\lambda \cdot \fval{C}^\lambda +  \fval{B}^\sigma \cdot \fval{C}^\lambda + \fval{B}^\lambda \cdot \fval{C}^\sigma +
		\fval{B}^\sigma \cdot \fval{C}^\sigma.
	\end{eqnarray*}
	By selecting from the last line all monomials of  length at least $n$, we get
	$$
	\fval{A}^\lambda = \fval{B}^\lambda \cdot \fval{C}^\lambda +  \fval{B}^\sigma \cdot \fval{C}^\lambda + \fval{B}^\lambda \cdot \fval{C}^\sigma + m_A .
	$$		 
	Applying to this equality the morphism $h$ and noting that $BC_{u,v}^{h(uv)}$ evaluates to $(\fval{B}(u)\fval{C}(v)) \cdot h(uv)$  
	shows that \eqref{eq:rhs-mult} is indeed the right definition for $\rhs_\lambda(A_\lambda)$.
	
	\medskip
	\noindent
	{\em Case 3.} $\rhs(A) = B \cdot C$ and $\bot\in\{\fval{B}^\sigma, \fval{B}^\lambda,\fval{C}^\sigma, \fval{C}^\lambda\}$.
	Then the corresponding terms on the right-hand side of \eqref{eq:rhs-mult} are omitted. More precisely 
	if $\fval{X}^\sigma = \bot$ ($X \in \{B,C\}$) then we omit in \eqref{eq:rhs-mult} the product involving $h(\fval{X}^\sigma)$ 
	as well as the semiring element $m_A$, and if $\fval{X}^\lambda = \bot$ ($X \in \{B,C\}$) then we omit in \eqref{eq:rhs-mult} the two products involving $X_\lambda$.
	The reader may also interpret $\bot$ as zero and then do the obvious simplifications in \eqref{eq:rhs-mult}  (but note that the semiring 
	$\N[R]$ has no additive zero element).
	For example, if $\fval{B}^\lambda = \fval{C}^\sigma = \bot$ and $\fval{B}^\sigma \neq \bot \neq \fval{C}^\lambda$
	then $\rhs_\lambda(A_\lambda) = h(\fval{B}^\sigma) C_\lambda$. Since $\fval{A}^\lambda \neq \bot$, at least one of the summands in 
	\eqref{eq:rhs-mult} remains.
	
	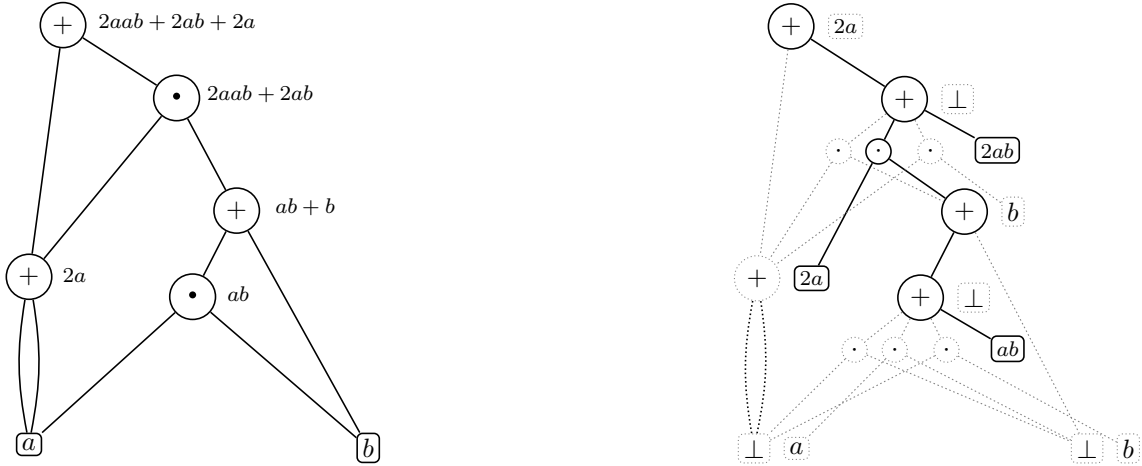
\begin{figure}[t]
		
		\centering
		\ifmac
		\begin{tikzpicture}
		\node[gate, label={[above right = -.5cm and 1.5cm of A]\footnotesize $2aab+2ab+2a$}] (A) {$+$};
		\node[gate, label={[above right = -.5cm and 1.1cm of A]\footnotesize $2aab+2ab$}, below right = 15pt and 30pt of A] (B) {\Huge $\cdot$};
		\node[gate, label={[above right = -.5cm and .9cm of A]\footnotesize $ab+b$}, below right = 30pt and 10pt of B] (C) {$+$};
		\node[gate, label={[above right = -.5cm and .6cm of A]\footnotesize $ab$}, below left = 20pt and 4pt of C] (D) {\Huge $\cdot$};
		\node[input, below left = 45pt and 50pt of D] (E) {$a$};
		\node[gate, label={[above right = -.5cm and 0.6cm of E2]\footnotesize $2a$}, above = 50pt of E]  (E2) {$+$};
		\node[input, below right = 45pt and 55pt of D] (F) {$b$};
		
		\draw[] (A) -- (E2);
		\draw[bend left = 10] (E2) edge (E);
		\draw[bend right = 10] (E2) edge (E);
		\draw[] (A) -- (B);
		\draw[] (B) -- (E2);
		\draw[] (B) -- (C);
		\draw[] (C) -- (D);
		\draw[] (C) -- (F);
		\draw[] (D) -- (E);
		\draw[] (D) -- (F);
		
		\end{tikzpicture}
		\else
		\begin{tikzpicture}
		\node[gate, label={[above right = -.5cm and .4cm of A]\footnotesize $2aab+2ab+2a$}] (A) {$+$};
		\node[gate, label={[above right = -.5cm and .4cm of A]\footnotesize $2aab+2ab$}, below right = 15pt and 30pt of A] (B) {\Huge $\cdot$};
		\node[gate, label={[above right = -.5cm and .4cm of A]\footnotesize $ab+b$}, below right = 30pt and 10pt of B] (C) {$+$};
		\node[gate, label={[above right = -.5cm and .4cm of A]\footnotesize $ab$}, below left = 20pt and 4pt of C] (D) {\Huge $\cdot$};
		\node[input, below left = 45pt and 50pt of D] (E) {$a$};
		\node[gate, label={[above right = -.5cm and .4cm of E2]\footnotesize $2a$}, above = 50pt of E]  (E2) {$+$};
		\node[input, below right = 45pt and 55pt of D] (F) {$b$};
		
		\draw[] (A) -- (E2);
		\draw[bend left = 10] (E2) edge (E);
		\draw[bend right = 10] (E2) edge (E);
		\draw[] (A) -- (B);
		\draw[] (B) -- (E2);
		\draw[] (B) -- (C);
		\draw[] (C) -- (D);
		\draw[] (C) -- (F);
		\draw[] (D) -- (E);
		\draw[] (D) -- (F);
		
		\end{tikzpicture}
		\fi
		\hfill
		\begin{tikzpicture}
		
		\node[gate] (A) {$+$};
		\node[gate, below right = 15pt and 30pt of A] (B) {$+$};
		\node[gate, below right = 30pt and 10pt of B] (C) {$+$};
		\node[gate, below left = 20pt and 4pt of C] (D) {$+$};
		\node[gate, bot, above = 50pt of E] (E2) {$+$};
		\node[input, right = 5pt of E2] (E2-s) {\footnotesize $2a$};
		\node[input, bot, below left = 45pt and 50pt of D] (E) {$\bot$};
		\node[input, bot, below right = 45pt and 50pt of D] (F) {$\bot$};
		
		\node[aux-gate, bot, below left = 10pt and 15pt of D] (D1) {$\cdot$};
		\node[aux-gate, bot, below left = 10pt and 0pt of D] (D2) {$\cdot$};
		\node[aux-gate, bot, below right = 10pt and 0pt of D] (D3) {$\cdot$};
		\node[input, below right = 8pt and 20pt of D] (D4) {\footnotesize $ab$};
		
		\node[aux-gate, bot, below left = 10pt and 15pt of B] (B1) {$\cdot$};
		\node[aux-gate, below left = 10pt and 0pt of B] (B2) {$\cdot$};
		\node[aux-gate, bot, below right = 10pt and 0pt of B] (B3) {$\cdot$};
		\node[input, below right = 8pt and 20pt of B] (B4) {\footnotesize $2ab$};
		
		\node[input, bot, right = 5pt of A] (A-s) {\footnotesize $2a$};
		\node[input, bot, right = 5pt of B] (B-s) {$\bot$};
		\node[input, bot, right = 5pt of C] (C-s) {$b$};
		\node[input,bot, right = 5pt of D] (D-s) {$\bot$};
		\node[input, bot, right = 5pt of E] (E-s) {$a$};
		\node[input, bot, right = 5pt of F] (F-s) {$b$};
		
		\draw[,bot] (A) -- (E2);
		\draw[] (A) -- (B);
		
		\draw[] (C) -- (D);
		\draw[,bot] (C) -- (F);
		
		\draw[,bot] (B) -- (B1);
		\draw[] (B) -- (B2);
		\draw[,bot] (B) -- (B3);
		\draw[] (B) -- (B4);
		
		\draw[,bot] (D) -- (D1);
		\draw[,bot] (D) -- (D2);
		\draw[,bot] (D) -- (D3);
		\draw[] (D) -- (D4);
		
		\draw[,bot] (D1) -- (E);
		\draw[,bot] (D1) -- (F);
		\draw[,bot] (D2) -- (E-s);
		\draw[,bot] (D2) -- (F);
		\draw[,bot] (D3) -- (E);
		\draw[,bot] (D3) -- (F-s);
		
		\draw[,bot] (B1) -- (E2);
		\draw[,bot] (B1) -- (C);
		\draw[] (B2) -- (E2-s);
		\draw[] (B2) -- (C);
		\draw[,bot] (B3) -- (E2);
		\draw[,bot] (B3) -- (C-s);
		
		\draw[bend left=10,bot] (E2) edge (E);
		\draw[bend right=10,bot] (E2) edge (E);
		
		\end{tikzpicture}
		\caption{The construction of the circuit $\C_\lambda$ in the proof of Lemma~\ref{lemma:monomlength}.}
		\label{fig-short-long}
	\end{figure}

	\medskip
	\noindent
	Figure~\ref{fig-short-long} shows an example of the above construction of the circuit $\C^\lambda$, where $n = 2$. The shaded parts are those parts that 
	are removed because they would yield zero terms. These are exactly those parts of the right-hand sides that are removed in the above 
	Case~3. To the right of each gate $X$, the value $\fval{X}^\sigma$ is written. Note that the output gate evaluates to $2aab + 2ab$
	which is indeed $(2aab + 2ab+2a)^\lambda$.

	\medskip
	\noindent
	{\em Step 5.}
	We now apply Lemma~\ref{lemma-normal-form} and transform in logspace  $\C_\lambda$ into a circuit $\C'$ in normal form.
	We have to argue that the circuit $\C'$ satisfies the conditions from Lemma~\ref{lemma:leaf-removement}
	for the ideal $I = \langle R^n \rangle$. 
	The input values of the circuit $\C_\lambda$ are elements from $\langle R^n \rangle$  
	(they occur as the $m_A$ in \eqref{eq:rhs-mult})
	and  the $h(\fval{X}^\sigma)$ for $X \in V$. Only the input values $h(\fval{X}^\sigma)$ can belong to $R \setminus \langle R^n\rangle$
	(they can also belong to $\langle R^n\rangle$). 
	The circuit $\C'$ is obtained from $\C_\lambda$ by (i) eliminating copy gates
	(that may arise from the above Case~1) and (ii) splitting up right-hand sides of the form \eqref{eq:rhs-mult} (or a simpler form, see Case~3).
	Note that in the circuit $\C'$ an input gate $Z$ with $\rhs_{\C'}(Z) \in R \setminus \langle R^n\rangle$ can only occur in right-hand sides of multiplication gates.
	Such a right-hand side must be of the form $B_\lambda \cdot Z$ or $Z \cdot C_\lambda$ (which is obtained from splitting up the expression in  \eqref{eq:rhs-mult}).
	Here, $B_\lambda$ and $C_\lambda$ are gates from the circuit $\C_\lambda$. But a
	gate $A_\lambda$ of the circuit $\C_\lambda$ cannot be transformed into an input gate of $\C'$ with 
	a right-hand side from $R \setminus \langle R^n\rangle$ (note that the values  $h(\fval{X}^\sigma)$  are ``guarded'' in \eqref{eq:rhs-mult} by
	multiplications with gates $Y_\lambda$).
	This shows that the conditions for Lemma~\ref{lemma:leaf-removement}
	are satisfied.
	
	Finally, Lemma~\ref{lemma:leaf-removement} allows us to transform in logspace the circuit $\C'$ into an equivalent normal form circuit $\D$ over the subsemiring 
	$\langle R^n \rangle = \langle RER \rangle$.
	This circuit $\D$ satisfies $[\D] = h(\fval{\C}^\lambda)$. We output this circuit together with the previously computed value $h(\fval{\C}^\sigma)$ (if this 
	value is not $\bot$). Then the output specification from the lemma is satisfied.
\end{proof}
The next lemma transforms a normal form circuit over $\langle RER \rangle$ into a type admitting circuit.

\begin{lemma}
	\label{lemma:FSF}
	Given a normal form circuit
	$\C=(V,A_0,\rhs)$ over $\langle RER \rangle$, one can compute in $\AC^0(\NL)$:
	\begin{itemize}
		\item  a type admitting circuit $\C' = (V',\rhs')$ (without output gate),
		\item a non-empty list of distinguished gates $A_1, \ldots, A_m \in V'$, where $m \leq |R|^4$, and
		\item elements $\ell_1, r_1, \ldots, \ell_m,r_m \in R$ such that $[\C] = \sum_{i=1}^m \ell_i [A_i]_{\C'} r_i$.
	\end{itemize}
\end{lemma}

\begin{proof}
	
	Let us interpret the circuit $\C = (V,A_0,\rhs)$ over the free semiring $\N[R]$.
	For each input gate $A$ we can write $[A]$ as $\sum_{i=1}^k s_i e_i^3 t_i$ for 
	a constant $k$ (that only depends on $R$), $s_i,t_i \in R$, $e_i \in E$ and
	redefine $\rhs(A) = \sum_{i=1}^k s_i e_i^3 t_i$ (a sum of $k$ monomials of length $5$).
	Thus for all $A \in V$ we have $\supp(\fval{A}_\C) \subseteq (RE^3R)^+ \subseteq RER^*ER$.
	
	Let us define for every inner gate $A$ of $\C$ the set
	$$
	P_A = \{ (s,e,f,t) \in R \times E \times E \times R \mid \supp(\fval{A}_\C) \cap s e R^* f t \neq \emptyset \}.
	$$
	Hence, $|P_A| \leq |R|^4$.
	We claim that the sets $P_A$ can be computed in $\AC^0(\NL)$. For this, note that 
	$(s,e,f,t) \in P_A$ if and only if (i) $e=f$ and there exists an input gate $C$ of $\C$ 
	such that the following conditions hold:
	\begin{itemize}
		\item $\rhs(C)$ contains the monomial $se^3t$.
		\item There is a path from $C$ to $A$ (possible the empty path) where all gates are addition gates.
	\end{itemize}
	or (ii) there exist input gates $C_1$, $C_2$ of $\C$, and a multiplication gate $B$ such that the following conditions hold:
	\begin{itemize}
		\item $\rhs(C_1)$ contains the monomial $se^3t'$ for some $t' \in R$.
		\item $\rhs(C_2)$ contains the monomial $s'f^3 t$ for some $s' \in R$.
		\item There is a path from $C_1$ to $B$ such that for every edge $(X,Y)$ along this path, where $Y$ is a multiplication gate,
		$\rhs(Y) = X \cdot Z$ for some gate $Z$.
		\item There is a path from $C_2$ to $B$ such that for every edge $(X,Y)$ along this path, where $Y$ is a multiplication gate,
		$\rhs(Y) = Z \cdot X$ for some gate $Z$.
		\item There is a path from $B$ to $A$ (possible the empty path) where except for $B$ all gates are addition gates. 
	\end{itemize}
	These conditions 
	can be checked in $\NL$.
	
	We next compute in logspace a new circuit $\mathcal{D}$ that contains for every  gate $A$ of $\C$ 
	and every tuple $(s,e,f,t) \in P_A$ a gate $A_{s,e,f,t}$ such that the following holds,
	where as usual $\fval{A_{s,e,f,t}}_\D$ denotes the evaluation of the gate $A_{s,e,f,t}$ 
	in the free semiring $\N[R]$ and $L(A,s,e,f,t) := \supp(\fval{A}_\C) \cap se R^* ft$ (which is non-empty since $(s,e,f,t) \in P_A$):
	\begin{equation} \label{eq-f_A,l,r}
	\fval{A_{s,e,f,t}}_\D = \sum_{w \in L(A,s,e,f,t)}  \fval{A}_\C(w) \cdot w
	\end{equation}
	Intuitively, we decompose the polynomial $\fval{A}_\C$ into several summands according to the first two and last two symbols in every monomial. 
	We define the rules of $\mathcal{D}$ as follows, where $A$ is a gate of $\C$: 
	
	\medskip
	\noindent
	{\em Case 1.} $\rhs(A) = \sum_{i=1}^k s_i e_i^3 t_i$.
	Then, we have $P_A = \{ (s_i,e_i,e_i,t_i) \mid 1 \le i \le k \}$ and we set 
	$$\rhs(A_{s_i,e_i,e_i,t_i}) = s_ie_i^3t_i.
	$$
	{\em Case 2.} $\rhs(A) = B \cdot C$ and $(s,e,f,t) \in P_A$.
	We set
	$$
	\rhs(A_{s,e,f,t}) = \sum_{(s,e,f',t') \in P_B} \sum_{(s',e',f,t) \in P_C} B_{s,e,f',t'} \cdot C_{s',e',f,t} .
	$$
	{\em Case 3.} $\rhs(A) = B + C$ and $(s,e,f,t) \in P_A$.
	We set
	$$\rhs(A_{s,e,f,t}) = B_{s,e,f,t} + C_{s,e,f,t}.
	$$
	With these right-hand sides, property \eqref{eq-f_A,l,r} is easy to verify.
	
	The idea of the last step is the following: Let $\overline{u} = (s,e,f,t) \in P_A$.
	Every non-commutative polynomial $\fval{A_{\overline{u}}}_\D$ has the property that each of its monomials
	starts with $see$ and ends with $fft$. By factoring out the common prefix $se$ and suffix $ft$, respectively,
	we can write $\fval{A_{\overline{u}}}_\D = s e g f t $, where $g \in e\,\N[R] f$ or $g=e$
	(the latter case occurs if $A$ is an input gate
	with right-hand side $se^3t$, in which case we have $e=f$). 
	We now construct in logspace a circuit $\C'$, which contains 
	gates $A'_{s,e,f,t}$ (where $A_{s,e,f,t}$ is a gate of $\D$ as above) such that in the free semiring $\N[R]$,
	$A'_{s,e,f,t}$ evaluates to the above polynomial $g$.
	We define the right-hand side of $A'_{s,e,f,t}$ again by a case distinction, where we use the right-hand sides
	for $\D$ that we defined in the Cases~1-3 above.
	
	\medskip
	\noindent
	{\em Case 1.} $e=f$ and  $\rhs(A_{s,e,e,t}) = se^3t$.
	Then, we set $\rhs(A'_{s,e,e,t}) = e$.
	
	\medskip
	\noindent
	{\em Case 2.} $\rhs(A_{s,e,f,t}) = \sum_{(s,e,f',t') \in P_B} \sum_{(s',e',f,t) \in P_C} B_{s,e,f',t'} \cdot C_{s',e',f,t}$.
	We set 
	\begin{equation} \label{eq-complex-rhs}
	\rhs(A'_{s,e,f,t}) = \sum_{(s,e,f',t') \in P_B} \sum_{(s',e',f,t) \in P_C} B'_{s,e,f',t'} (f' t' s' e') C_{s',e',f,t} .
	\end{equation}
	{\em Case 3.} $\rhs(A_{s,e,f,t}) = B_{s,e,f,t} + C_{s,e,f,t}$.
	We set 
	$$
	\rhs(A'_{s,e,f,t}) = B'_{s,e,f,t} + C'_{s,e,f,t}.
	$$
	It is now straightforward to verify that for every gate $A$ of $\C$ we have:
	$$
	\fval{A}_{\C} = \sum_{(s,e,f,t) \in P_A} \fval{A_{s,e,f,t}}_{\D} = \sum_{(s,e,f,t) \in P_A} s e \fval{A'_{s,e,f,t}}_{\C'} f t 
	$$
	Hence, if we evaluate the circuits $\C$, $\D$, and $\C'$ in the semiring $R$ we get
	\begin{eqnarray*}
		[A]_{\C} =  \sum_{(s,e,f,t) \in P_A} [A_{s,e,f,t}]_{\D} & = & \sum_{(s,e,f,t) \in P_A} s e [A'_{s,e,f,t}]_{\C'} f t  .
	\end{eqnarray*}
	Note that $[A'_{s,e,f,t}]_{\C'} \in eRf$, which holds, since $[A'_{s,e,f,t}]_{\C'} = h(  \fval{A'_{s,e,f,t}}_{\C'} )$ and every
	monomial of $\fval{A'_{s,e,f,t}}_{\C'}$ starts with $e$ and ends with $f$.
	Moreover, every input value of the circuit $\C'$ is from $ERE$: These are the elements $e$ (in Case 1) and $f' t' s' e'$ (in Case 2). 
	
	Note that $\C'$ admits a type function; We set $\mathsf{type}(A'_{s,e,f,t}) = (e,f)$. Moreover, using Lemma~\ref{lemma-normal-form} we 
	transform  $\C'$ in logspace into normal form by splitting up right-hand sides
	of the form \eqref{eq-complex-rhs}. Thereby we extend the type-mapping to the new gates that are
	introduced. For instance, if we introduce a gate with right-hand side $B'_{s,e,f',t'} (f' t' s' e')$ (which occurs
	in \eqref{eq-complex-rhs}), then this gate gets the type $(e,e')$, and the gate that computes (in two steps)
	$B'_{s,e,f',t'} (f' t' s' e') C_{s',e',f,t}$ gets the type $(e,f)$. This ensures that the three conditions from Definition~\ref{def-type} 
	are satisfied.
\end{proof}
Combining Lemma~\ref{lemma:monomlength} and \ref{lemma:FSF} immediately yields Proposition~\ref{thm:step-1}.

\subsection{Step 2: A parallel evaluation algorithm for type admitting circuits} \label{sec-rank}

In this section we prove Proposition~\ref{thm:step-2}.
We present a parallel evaluation algorithm for type admitting circuits.
This algorithm terminates after at most $|R|$ rounds, if $R$ has a so called rank-function,
which we define first. As before, let $E = E(R)$. 
 
\begin{definition} \label{def-rank}
	We call a function $\rank: R \to \N \setminus \{0\}$ a {\em rank-function} for $R$
	if it satisfies the following conditions for all $a,b \in R$:
	\begin{enumerate}
		\item $\rank(a) \le \rank(a+b)$
		\item $\rank(a), \rank(b) \le \rank(a \cdot b)$
		\item If $a,b \in eRf$ for some $e,f \in E$ and $\rank(a) = \rank(a+b)$, then $a = a+b$.
	\end{enumerate}
\end{definition}
Note that if $\mR$ is a monoid, then one can choose $e=1=f$ in the third condition in Definition~\ref{def-rank},
which is therefore equivalent to:
If $\rank(a) = \rank(a+b)$ for $a,b \in R$, then $a = a+b$.

\begin{example}[Example~\ref{ex-power-semigroup} continued] \label{ex-rank}
Let $G$ be a finite group and consider the semiring $\mathcal{P}(G)$.
One can verify that the function $A \mapsto |A|$, where $\emptyset \neq A \subseteq G$, is a rank-function for $\mathcal{P}(G)$.
On the other hand, if $S$ is a finite semigroup, which is not a group, then $S$ cannot be cancellative.
Assume that $ab = ac$ for $a,b,c \in S$ with $b \neq c$. Then $\{a\} \cdot \{ b,c\} = \{ab\}$. This shows that
the function $A \mapsto |A|$ is not a rank-function for $\mathcal{P}(S)$.
\end{example}

\begin{theorem} 
	\label{thm:rank}
	If the finite semiring $R$ has a rank-function $\rank$, 
	then the restriction of $\CEP(R)$ to type admitting circuits belongs to $\AC^0(\NL, \CEP(\aR), \CEP(\mR))$.
\end{theorem}

\begin{proof}
	Let $\C = (V,A_0,\rhs)$ be a circuit with the type function $\mathsf{type}$.
	We present an algorithm which partially evaluates the circuit in a constant number of phases, where each phase can be carried
	out in $\AC^0(\NL, \CEP(\aR), \CEP(\mR))$ and
	the following invariant is preserved:
	\begin{inv}
		After phase $k$ all gates $A$ with $\rank([A]_\C) \le k$ are evaluated, i.e., are input gates.
	\end{inv}
	\noindent
	In the beginning, i.e., for $k=0$, the invariant clearly holds (since $0$ is not in the range of the rank-function).
	After $\mathsf{max} \{ \rank(a) \mid a \in R \}$ (which is a constant) many phases the output gate $A_0$ is evaluated.
	We present phase $k$ of the algorithm, assuming that the invariant holds after phase $k-1$.
	Thus, all gates $A$ with $\rank([A]_\C) < k$ of the current circuit $\C$ are input gates.
	The goal of phase $k$ is to evaluate all gates $A$ with $\rank([A]_\C) = k$.
	For this, we proceed in two steps:
	
	\medskip
	\noindent
	{\em Step 1.}
	 As a first step the algorithm evaluates all subcircuits that only contain addition and input gates.
	This maintains the invariant and is possible in $\AC^0(\NL, \CEP(\aR))$.
	After this step, every addition-gate $A$ has at least one inner input gate,
	which we denote by $\mathsf{inner}(A)$ (if both input gates are inner gates, then choose one arbitrarily).
	The $\NL$-oracle access is needed to compute 
	the set of all gates $A$ for which no multiplication gate $B \leq_{\C} A$ exists.
	
	\medskip
	\noindent
	{\em Step 2.}
	Define the multiplicative circuit $\C' = (V,A_0,\rhs')$ by
	 \begin{equation} \label{def-C'-rhs}
           \rhs'(A) = \begin{cases}
           \mathsf{inner}(A) & \text{ if $A$  is an addition-gate,} \\
           \rhs(A)                 & \text{ if  $A$ is a multiplication gate or input gate.}
           \end{cases}
          \end{equation}
       	The circuit $\C'$ can be brought into normal form by Lemma~\ref{lemma-normal-form} and then
	evaluated using the oracle for $\CEP(\mR)$.
	A gate $A \in V$ is called {\em locally correct} if (i) $A$ is an input gate or multiplication gate of $\C$, or (ii)
	$A$ is an addition gate of $\C$ with $\rhs(A) = B+C$ and $[A]_{\C'} = [B]_{\C'} + [C]_{\C'}$.
	We compute the set
	\[
		W = \{ A \in V \mid B \text{ is locally correct for all gates } B \text{ with } B \le_\C A \}
	\]
	in $\AC^0(\NL)$. A simple induction shows that for all $A \in W$ we have $[A]_\C = [A]_{\C'}$. Hence
	we can set $\rhs(A) = [A]_{\C'}$ for all $A \in W$. This concludes phase $k$ of the algorithm.
	
	\medskip
	\noindent
	To prove that the invariant still holds after phase $k$,
	we show that for each gate $A \in V$ with $\rank([A]_{\C}) \le k$ we have $A \in W$.
	This is shown by induction over the depth of $A$ in $\C$. Assume that $\rank([A]_{\C}) \le k$.
	By the first two conditions from Definition~\ref{def-rank}, all gates $B <_\C A$ satisfy $\rank([B]_\C) \le k$.
	Thus, the induction hypothesis yields $B \in W$ and hence $[B]_\C = [B]_{\C'}$ for all gates $B <_\C A$.
	
         It remains to show that $A$ is locally correct, which is clear if $A$ is an input gate or a multiplication gate.
	So assume that $\rhs(A) = B+C$ where $B = \mathsf{inner}(A)$, which implies $[A]_{\C'} = [B]_{\C'}$ by 
	\eqref{def-C'-rhs}. 
	Since $B$ is an inner gate, which is not evaluated after phase $k-1$,
	it holds that $\rank([B]_\C) \ge k$ and therefore $\rank([A]_\C) = \rank([B]_\C) = k$.
	By Definition~\ref{def-type} there exist idempotents $e,f \in E$  with 
	$\mathsf{type}(B) = \mathsf{type}(C) = (e,f)$ and thus $[B]_\C, [C]_\C \in eRf$.
	The third condition from Definition~\ref{def-rank} implies that $[A]_\C = [B]_\C + [C]_\C = [B]_\C$.
	We get $$[A]_{\C'} = [B]_{\C'} = [B]_\C = [A]_\C = [B]_\C + [C]_\C = [B]_{\C'} + [C]_{\C'}.$$
	Therefore $A$ is locally correct.
\end{proof}

\begin{example}[Example~\ref{ex-power-semigroup} continued]
        Figure~\ref{fig-rank} shows a circuit $\C$ over
	the power semiring $\mathcal{P}(G)$ of the group $G = (\Z_5,+)$.
	Recall from Example~\ref{ex-rank} that the function $A \mapsto |A|$ is a rank function for $\mathcal{P}(G)$.
	We illustrate one phase of the algorithm.
	All gates $A$ with $\rank([A]) < 3$ are evaluated in the circuit $\C$ shown in (a).
	The goal is to evaluate all gates $A$ with $\rank([A]) = 3$.
	The first step would be to evaluate maximal $\cup$-circuits, which is already done.
	In the second step the circuit $\C'$ (shown in (b)) from the proof of Lemma~\ref{thm:rank} is computed and evaluated using the oracle
	for $\CEP(\Z_5,+)$. The dotted wires do not belong to the circuit $\C'$.
	All locally correct gates are shaded. Note that the output gate is locally correct
	but its right child is not locally correct.
	All other shaded gates form a downwards closed set, which is the set $W$ from the proof.
	These gates can be evaluated such that in the resulting circuit (shown in (c))
	all gates which evaluate to elements of rank 3 are evaluated.
\end{example}
%%%%
\iftrue
%%%%
\begin{figure}[t]
	\centering
	\ifmac
	\begin{minipage}[b]{.32\linewidth}
		\scalebox{0.8}{
\begin{tikzpicture}
\node[gate, label={[above = 0cm of A]\scriptsize $\{0,1,2,3,4\}$}] (A) {$\cup$};
\node[gate, label={[above left = 0cm and .3cm of B]\scriptsize $\{1,2,3\}$}, below left = 15pt and 15pt of A] (B) {$\cup$};
\node[gate, label={[above right = 0cm and .5cm of C]\scriptsize $\{0,1,2,3,4\}$}, below right = 15pt and 15pt of A] (C) {$\cup$};
\node[gate, label={[above left = 0cm and .3cm of D]\scriptsize $\{1,2,3\}$}, below left = 15pt and 5pt of B] (D) {$+$};
\node[gate, label={[above right = 0cm and .5cm of E]\scriptsize $\{1,2,3,4\}$}, below right = 15pt and 5pt of C] (E) {$+$};
\node[input, below left = 15pt and -10pt of D] (F) {\footnotesize $\{0,1\}$};
\node[input, below = 67pt of A] (G) {\footnotesize $\{1,2\}$};
\node[input, below right = 15pt and -10pt of E] (H) {\footnotesize $\{0,2\}$};

\draw (A) -- (B);
\draw (A) -- (C);
\draw (B) -- (D);
\draw (C) -- (E);
\draw (D) -- (F);
\draw (D) -- (G);
\draw (E) -- (G);
\draw (E) -- (H);
\draw (C) -- (F);
\draw (B) -- (G);

\end{tikzpicture}
		}
	\end{minipage}
	\hfill
	\begin{minipage}[b]{.32\linewidth}
		\scalebox{0.8}{
 \begin{tikzpicture}
 \node[gate, marked, label={[above = 0cm of A]\scriptsize $\{1,2,3,4\}$}] (A) {$\cup$};
 \node[gate, marked, label={[above left = 0cm and .3cm of B]\scriptsize $\{1,2,3\}$}, below left = 15pt and 15pt of A] (B) {$\cup$};
 \node[gate, label={[above right = 0cm and .3cm of C]\scriptsize $\{1,2,3,4\}$}, below right = 15pt and 15pt of A] (C) {$\cup$};
 \node[gate, marked, label={[above left = 0cm and .3cm of D]\scriptsize $\{1,2,3\}$}, below left = 15pt and 5pt of B] (D) {$+$};
 \node[gate, marked, label={[above right = 0cm and .5cm of E]\scriptsize $\{1,2,3,4\}$}, below right = 15pt and 5pt of C] (E) {$+$};
 \node[input, marked, below left = 15pt and -10pt of D] (F) {\footnotesize $\{0,1\}$};
 \node[input, marked, below = 67pt of A] (G) {\footnotesize $\{1,2\}$};
 \node[input, marked, below right = 15pt and -10pt of E] (H) {\footnotesize $\{0,2\}$};
 
 \draw[dotted] (A) -- (B);
 \draw (A) -- (C);
 \draw (B) -- (D);
 \draw (C) -- (E);
 \draw (D) -- (F);
 \draw (D) -- (G);
 \draw (E) -- (G);
 \draw (E) -- (H);
 \draw[dotted] (C) -- (F);
 \draw[dotted] (B) -- (G);
 
 \end{tikzpicture}

		}
	\end{minipage}
	\hfill
	\begin{minipage}[b]{.32\linewidth}
		\scalebox{0.8}{
               \begin{tikzpicture}
               \node[gate] (A) {$\cup$};
               \node[input, below left = 15pt and 0pt of A] (B) {\footnotesize $\{1,2,3\}$};
               \node[gate, below right = 15pt and 10pt of A] (C) {$\cup$};
               \node[input, below left = 15pt and -15pt of B] (D) {\footnotesize $\{1,2,3\}$};
               \node[input, below right = 15pt and -10pt of C] (E) {\footnotesize $\{1,2,3,4\}$};
               \node[input, below left = 15pt and -20pt of D] (F) {\footnotesize $\{0,1\}$};
               \node[input, below = 69pt of A] (G) {\footnotesize $\{1,2\}$};
               \node[input, below right = 15pt and -20pt of E] (H) {\footnotesize $\{0,2\}$};
               
               \draw (A) -- (B);
               \draw (A) -- (C);
               \draw (C) -- (E);
               \draw (C) -- (F);
               
               \end{tikzpicture}

		}
	\end{minipage}
	\else
	\begin{minipage}[b]{.32\linewidth}
		\scalebox{0.8}{
                        \begin{tikzpicture}
                        \node[gate, label={[above = 0cm of A]\scriptsize $\{0,1,2,3,4\}$}] (A) {$\cup$};
                        \node[gate, label={[above left = 0cm and -.3cm of B]\scriptsize $\{1,2,3\}$}, below left = 15pt and 15pt of A] (B) {$\cup$};
                        \node[gate, label={[above right = 0cm and -.4cm of C]\scriptsize $\{0,1,2,3,4\}$}, below right = 15pt and 15pt of A] (C) {$\cup$};
                        \node[gate, label={[above left = 0cm and -.3cm of D]\scriptsize $\{1,2,3\}$}, below left = 15pt and 5pt of B] (D) {$+$};
                        \node[gate, label={[above right = 0cm and -.3cm of E]\scriptsize $\{1,2,3,4\}$}, below right = 15pt and 5pt of C] (E) {$+$};
                        \node[input, below left = 15pt and -10pt of D] (F) {\footnotesize $\{0,1\}$};
                        \node[input, below = 67pt of A] (G) {\footnotesize $\{1,2\}$};
                        \node[input, below right = 15pt and -10pt of E] (H) {\footnotesize $\{0,2\}$};
                        
                        \draw (A) -- (B);
                        \draw (A) -- (C);
                        \draw (B) -- (D);
                        \draw (C) -- (E);
                        \draw (D) -- (F);
                        \draw (D) -- (G);
                        \draw (E) -- (G);
                        \draw (E) -- (H);
                        \draw (C) -- (F);
                        \draw (B) -- (G);
                        
                        \end{tikzpicture}
			
		}
	\end{minipage}
	\hfill
	\begin{minipage}[b]{.32\linewidth}
		\scalebox{0.8}{
			\begin{tikzpicture}
			\node[gate, marked, label={[above = 0cm of A]\scriptsize $\{1,2,3,4\}$}] (A) {$\cup$};
			\node[gate, marked, label={[above left = 0cm and -.3cm of B]\scriptsize $\{1,2,3\}$}, below left = 15pt and 15pt of A] (B) {$\cup$};
			\node[gate, label={[above right = 0cm and -.3cm of C]\scriptsize $\{1,2,3,4\}$}, below right = 15pt and 15pt of A] (C) {$\cup$};
			\node[gate, marked, label={[above left = 0cm and -.2cm of D]\scriptsize $\{1,2,3\}$}, below left = 15pt and 5pt of B] (D) {$+$};
			\node[gate, marked, label={[above right = 0cm and -.3cm of E]\scriptsize $\{1,2,3,4\}$}, below right = 15pt and 5pt of C] (E) {$+$};
			\node[input, marked, below left = 15pt and -10pt of D] (F) {\footnotesize $\{0,1\}$};
			\node[input, marked, below = 67pt of A] (G) {\footnotesize $\{1,2\}$};
			\node[input, marked, below right = 15pt and -10pt of E] (H) {\footnotesize $\{0,2\}$};
			
			\draw[dotted] (A) -- (B);
			\draw (A) -- (C);
			\draw (B) -- (D);
			\draw (C) -- (E);
			\draw (D) -- (F);
			\draw (D) -- (G);
			\draw (E) -- (G);
			\draw (E) -- (H);
			\draw[dotted] (C) -- (F);
			\draw[dotted] (B) -- (G);
			
			\end{tikzpicture}
		}
	\end{minipage}
	\hfill
	\begin{minipage}[b]{.32\linewidth}
		\scalebox{0.8}{
			\begin{tikzpicture}
			\node[gate] (A) {$\cup$};
			\node[input, below left = 15pt and 0pt of A] (B) {\footnotesize $\{1,2,3\}$};
			\node[gate, below right = 15pt and 10pt of A] (C) {$\cup$};
			\node[input, below left = 15pt and -15pt of B] (D) {\footnotesize $\{1,2,3\}$};
			\node[input, below right = 15pt and -10pt of C] (E) {\footnotesize $\{1,2,3,4\}$};
			\node[input, below left = 15pt and -20pt of D] (F) {\footnotesize $\{0,1\}$};
			\node[input, below = 69pt of A] (G) {\footnotesize $\{1,2\}$};
			\node[input, below right = 15pt and -20pt of E] (H) {\footnotesize $\{0,2\}$};
			
			\draw (A) -- (B);
			\draw (A) -- (C);
			\draw (C) -- (E);
			\draw (C) -- (F);
			
			\end{tikzpicture}
		}
	\end{minipage}
	\fi
	\caption{The parallel evaluation algorithm over the power semiring $\mathcal{P}(\Z_5)$.}
	\label{fig-rank}
\end{figure}
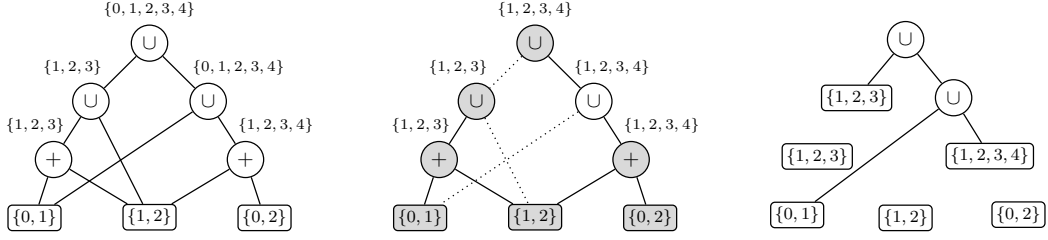
%%%%
\fi
For the proof of Proposition~\ref{thm:step-2},
it remains to show that every finite $\{0,1\}$-free semiring has a rank-function. 

\begin{lemma}
	\label{lem:char=(n,1)}
	Let $R$ be $\{0,1\}$-free. Let $e,f \in R$ such that
	\begin{itemize}
		\item $ef = fe = f$,
		\item $e^2 = e$,
		\item $f^2 = f+f=f$.
	\end{itemize}
	Then $e+f=f$.
\end{lemma}

\begin{proof}
       With $f=0$ and $e+f=1$ all equations from Lemma~\ref{lemma-0-1-free} (point 4) hold. Hence, we must have $e+f=f$.
 %	The set $\{f, f+e, f+2e, \dots \}$ is a subsemiring with zero element $f$ and one element $f+e$,
%	which isomorphic to some $B(n,d)$. Since $R$ is $\{0,1\}$-free, we must have $(n,d) = (0,1)$,
%	hence $e+f=f$.
\end{proof}

\begin{lemma}
	If the finite semiring $R$ is $\{0,1\}$-free, then $R$ has a rank-function.
\end{lemma}

\begin{proof}
	For $a,b \in R$ we define $a \preceq b$ if $b$ can be obtained from $a$ by iterated additions 
	and left- and right-multiplications of elements from $R$. This is equivalent to the following condition:
	\[
	\exists \ell, r, c \in R :
	b = \ell a r + c \text{ (where each of the elements $\ell, r, c$ can be also missing)}
	\]
	Since $\preceq$ is a preorder on $R$, there is a function $\rank: R \to \N\setminus \{0\}$	
	such that for all $a,b \in R$ we have
	\begin{itemize}
		\item $\rank(a) = \rank(b)$ iff $a \preceq b$ and $b \preceq a$,
		\item $\rank(a) \le \rank(b)$ if $a \preceq b$.
	\end{itemize}
	We claim that $\rank$ satisfies the conditions of Definition~\ref{def-rank}. The first two conditions are clear,
	since $a \preceq a+b$ and $a,b \preceq ab$.
	For the third condition, let $e,f \in E$, $a,b \in eRf$ such that $\rank(a+b) = \rank(a)$, which is equivalent to
	$a+b \preceq a$.
	Assume that $a = \ell (a+b) r + c = \ell a r + \ell b r + c$ for some $\ell, r, c \in R$ (the case without $c$ can be handled in the same way).
	Since $a = eaf$	and $b = ebf$, we have $a = \ell e(a+b)f r + c$ and hence we can assume that $\ell$ and $r$ are not missing.
	By multiplying with $e$ from the left and $f$ from the right we get $a = (e \ell e)(a+b)(f r f) + (ecf)$,
	so we can assume that $\ell = e \ell e$ and $r = frf$.
	After $m$ repeated applications of $a= \ell a r + \ell b r + c$ we obtain
	\begin{equation} \label{eq:a=...}
	a = \ell^mar^m + \sum_{i=1}^m \ell^ibr^i + \sum_{i=0}^{m-1} \ell^icr^i.
	\end{equation}
	Let $n \ge 1$ such that $n x$ is additively idempotent and $x^n$ is multiplicatively idempotent for all $x \in R$.
	Hence $n x^n$ is both additively and multiplicatively idempotent for all $x \in R$.
	If we choose $m = n^2$, the right hand side of \eqref{eq:a=...} contains the partial sum
	$\sum_{i=1}^n \ell^{i n} b r^{i n}$.
	Furthermore, $e(n \ell^n) = (n \ell^n)e = n \ell^n$ and $f(n r^n) = (n r^n)f = n r^n$. Therefore, Lemma~\ref{lem:char=(n,1)} implies
	that $n \ell^n = n\ell^n + e$ and $n r^n = n r^n + f$, and hence:
	\begin{align*}
	\sum_{i=1}^n \ell^{in} b r^{in} &= n (\ell^{n} b r^{n}) = n^2 (\ell^{n} b r^{n}) = 
	(n\ell^n) b (nr^n)  =    (n\ell^n + e) b (nr^n)  \\       
	& = (n\ell^n) b (nr^n) + eb(nr^n) =  (n\ell^n) b (nr^n) + eb(nr^n+f)  
	\\
	&= (n\ell^n) b (nr^n) + eb(nr^n) + ebf  = \left(\sum_{i=1}^n \ell^{i n} b r^{i n}\right) + b.
	\end{align*}
	Thus, we can replace in \eqref{eq:a=...} the partial sum $\sum_{i=1}^n \ell^{i n} b r^{i n}$ by
	$\sum_{i=1}^n \ell^{i n} b r^{i n} + b$, 
	which proves that $a = a + b$.
\end{proof}

\section{An application to formal language theory} \label{sec-intersection}

We present an application of our complexity results for circuit evaluation to formal language theory.
Recall that a {\em context-free grammar} (over $\Sigma$) is a tuple $\G = (V,\Sigma,S,P)$ consisting of a finite set of variables $V$,
a finite alphabet $\Sigma$, a start variable $S \in V$ and a set $P$ of productions $A \to \alpha$ where $A \in V$ and $\alpha \in (V \cup \Sigma)^*$.
We write $L_\G(A)$ for the language of $A$, i.e. the set of words $w \in \Sigma^*$ which can be derived from $A$ using the productions in $P$,
and write $L(\G)$ for $L_\G(S)$. 
Every circuit over the free monoid $\Sigma^*$ can be seen as a context-free grammar producing exactly one word. Such a circuit
is also called a {\em straight-line program}, briefly SLP. It is a context-free grammar $\mathcal{H} = (V,\Sigma,S,P)$ that
contains for every variable $A \in V$ exactly one rule of the form $A \to \alpha$. Moreover, $\mathcal{H}$ is acyclic, i.e., 
there is no non-empty derivation from a variable $A$ to a word containing $A$. We denote with $\val_{\mathcal{H}}(A)$
the unique word in the language $L_{\mathcal{H}}(A)$. Moreover, let $\val(\mathcal{H}) = \val_{\mathcal{H}}(S)$.

Given an alphabet $\Sigma$ and a language $L \subseteq \Sigma^*$, the {\em intersection non-emptiness problem for $L$},
denoted by \textsf{CFG-IP}$(L,\Sigma)$, is the following decision problem:

\problem{A context-free grammar $\G$ over $\Sigma$}{Does $L(\G) \cap L \neq \emptyset$ hold?}

For every regular language $L$, this problem is solvable in polynomial-time by
constructing a context-free grammar for $L(\G) \cap L$ from the given grammar $\G$ and a finite automaton for $L$.
The constructed grammar then has to be tested for emptiness, which is possible in polynomial time.
However, testing emptiness of a given context-free language is $\Ptime$-complete \cite{JonesL76}.
An easy reduction shows that the problem \textsf{CFG-IP}$(L,\Sigma)$ is $\Ptime$-complete for any non-empty language $L$.

\begin{theorem} \label{thm-p-complete-intersection}
	For every non-empty language $L \subseteq \Sigma^*$, $\mathsf{CFG}$-$\mathsf{IP}(L,\Sigma)$ is $\Ptime$-complete.
\end{theorem}

\begin{proof}
	Let $\G = (V,\Sigma,S,P)$ be a context-free grammar. We reduce emptiness of $\G$ to the intersection non-emptiness problem as follows.
	Let $X \not\in V$ be a new variable. We replace all occurrences of terminal symbols in productions of $\G$ by
	$X$ and then add the rules $X \to \varepsilon$ and $X \to a X$ for all $a \in \Sigma$ (thus, $X$ produces $\Sigma^*$).
	Observe that the new grammar $\G'$ satisfies $L(\G) \neq \emptyset$ if and only if $L(\G') = \emptyset$.
	Further, $L(\G')$ is either $\emptyset$ or $\Sigma^*$.
	Hence, $L(\G) \neq \emptyset$ if and only if $L(\G') \cap L \neq \emptyset$.
	Clearly, the reduction can be performed in logspace.
\end{proof}
By Theorem~\ref{thm-p-complete-intersection} we have to put some restriction on context-free grammars in order
to get $\NC$-algorithms for the intersection non-emptiness problem. It turns out that productivity of all variables is the right assumption.
Thus, we require that $L_\G(A) \neq \emptyset$ for all $A \in V$. In order to avoid a promise problem (testing productivity
of a variable is $\Ptime$-complete) we add to the input grammar $\G = (V,\Sigma,S,P)$ an SLP $\mathcal{H} = (V,\Sigma,S,R)$
which  {\em uniformizes} $\G$ in the sense that $R$ contains for every variable $A \in V$ exactly one rule $(A \to \alpha) \in P$.
Hence, the word $\val_{\mathcal{H}}(A) \in L_{\G}(A)$ is a witness for $L_\G(A) \neq \emptyset$.

\begin{example}
Here is a context-free grammar, where the underlined productions form a uniformizing SLP:
$$
S \to SS, \ S \to aSb, \ \underline{S \to A}, \  A \to aA, \ \underline{A \to B}, \ B \to bB, \ \underline{B \to b}
$$
\end{example}

We study the following decision problem \textsf{PCFG-IP}$(L,\Sigma)$ in the rest of this section:

\problem{A productive context-free grammar $\G$ over $\Sigma$ and a uniformizing SLP $\mathcal{H}$ for $\G$.}{Does 
$L(\G) \cap L \neq \emptyset$ hold?}

The goal of this section is to classify regular languages $L$ by the complexity of \textsf{PCFG-IP}$(L,\Sigma)$.

\subsection{Reduction to circuit evaluation}

In the following we prove that \textsf{PCFG-IP}$(L,\Sigma)$ is equivalent (with respect to constant depth
reductions) to the circuit evaluation problem for a suitable finite semiring that is derived from $L$.

We start with a few standard notations from algebraic language theory.
A language $L \subseteq \Sigma^*$ is {\em recognized} by a monoid $M$
if there exists a homomorphism $h: \Sigma^* \to M$ such that $h^{-1}(F) = L$ for some $F \subseteq M$. 
It is known that a language is regular if and only if it is recognized by a finite monoid.
The {\em syntactic congruence} $\equiv_L$ is the equivalence relation on $\Sigma^*$ that is defined by
$u \equiv_L v$ ($u,v \in \Sigma^*$) if the following holds: $\forall x,y \in \Sigma^* : xuy \in L \Leftrightarrow xvy \in L$.
It is indeed a congruence relation on the free monoid $\Sigma^*$.
The quotient monoid $\Sigma^*/{\equiv_L}$ is the smallest monoid which recognizes $L$; it is called 
the {\em syntactic monoid of $L$}.
From now on we fix a language $L \subseteq\Sigma^*$, a surjective homomorphism $h: \Sigma^* \to M$ onto the syntactic monoid $M$ of $L$
and a set $F \subseteq M$ satisfying $h^{-1}(F) = L$.
As a variation of computation problem $\CEP(\mathcal{P}(M))$, we define the decision problem $\CEP(\mathcal{P}(M),F)$:

\problem{A circuit over $\mathcal{P}(M)$}{Does $[\C] \cap F \neq \emptyset$ hold?}

\begin{lemma} \label{lemma-PCFG-CEP}
	$\mathsf{PCFG}$-$\mathsf{IP}(L,\Sigma)$ is equivalent to $\CEP(\mathcal{P}(M),F)$ with respect to constant depth reductions.
\end{lemma}

\begin{proof} We first reduce $\mathsf{PCFG}$-$\mathsf{IP}(L,\Sigma)$  to $\CEP(\mathcal{P}(M),F)$. 
Let $\G = (V,\Sigma,S,P)$ be a productive context-free grammar
	and let $\mathcal{H} = (V,\Sigma,S,R)$ be a uniformizing SLP for $\G$.
	To decide whether $L(\G) \cap L \neq \emptyset$, we construct a circuit whose gates compute all 
	sets $X_A = h(L_{\G}(A)) \in \mathcal{P}(M)$
	for $A \in V$. Then we test whether $X_S$ intersects $F$.
	
	In preparation we compute from $\mathcal{H}$ a circuit whose gates evaluate to 
	the singleton sets $X_A^{(0)} := \{ h(\val_{\mathcal{H}}(A)) \}$ for $A \in V$.
	Every production $(A \to \alpha_0 A_1 \alpha_1 \cdots A_k \alpha_k) \in R$ 
	with $A_1, \dots, A_k \in V$ and $\alpha_0, \dots, \alpha_k \in \Sigma^*$
	is translated to the definition
	\[
		\rhs(A) = \{ h(\alpha_0) \} \cdot A_1 \cdot \{ h(\alpha_1) \} \cdots  A_k \cdot \{ h(\alpha_k) \}
	\]
	in the circuit.
	Now the tuple $(X_A)_{A \in V}$ is the least fixed-point of the following monotone operator $\mu$:
	\begin{gather}
		\mu:  \left(2^M \right)^{|V|} \to \left(2^M \right)^{|V|} \\
		\mu((Y_A)_{A \in V}) =  \big(Y_A \cup \, \bigcup h(\alpha_0) Y_{A_1} h(\alpha_1) \cdots Y_{A_k} h(\alpha_k)\big)_{A \in V} 
		     \label{eq:hom-prod}
	\end{gather}
	where the union in \eqref{eq:hom-prod} ranges over all productions $A \to \alpha_0 A_1 \alpha_1 \cdots A_k \alpha_k \in P$
	for $A_1, \dots, A_k \in V$ and $\alpha_0, \dots, \alpha_k \in \Sigma^*$.
	The smallest fixpoint of $\mu$ can be computed by the fixed-point iteration
	\begin{equation}
		Y_A^{(0)} = \emptyset, \quad Y_A^{(n+1)} = \mu((Y_A^{(n)})_{A \in V}) \label{eq:fp-it}
	\end{equation}
	which reaches the least fixed-point after at most $|V| \cdot |M|$ steps. Equation \eqref{eq:fp-it}
	gives rise to an $\AC^0$-computable circuit over the semiring $(2^M,\cup,\cdot)$
	computing $X_S = h(L(\G))$.
	Since we disallow the empty set in $\mathcal{P}(M)$, we instead initialize the fixed-point iteration in \eqref{eq:fp-it}
	with the non-empty subsets $X_A^{(0)} \subseteq X_A$. This yields a circuit over $\mathcal{P}(M)$ for $X_S$.
	
	Let us now reduce $\CEP(\mathcal{P}(M),F)$ to $\mathsf{PCFG}$-$\mathsf{IP}(L,\Sigma)$.
	Let $\C = (V,A_0,\rhs)$ be a circuit over $\mathcal{P}(M)$. We define a grammar $\G = (V,\Sigma,A_0,P)$ as follows:
	\begin{itemize}
		\item If $\rhs(A) = \{m_1,\ldots, m_k \} \in \mathcal{P}(M)$, add the rules 
		$A \to w_i $ to $P$ ($1 \leq i \leq k$) where $w_i \in \Sigma^*$ is any word with $h(w_i) = m_i$.
		\item If $\rhs(A) = B \cup C$, add the rules $A \to B$ and $A \to C$ to $P$.
		\item If $\rhs(A) = B \cdot C$, add the rules $A \to BC$ to $P$.
	\end{itemize}
	Then every gate $A \in V$ evaluates to $h(L_\G(A))$.
	In particular, we have $h(L(\G)) = [\C]$.
	Therefore, $[\C] \cap F \neq \emptyset$ if and only if $L(\G) \cap L \neq \emptyset$.
\end{proof}

Now clearly $\CEP(\mathcal{P}(M),F)$ is logspace reducible to $\CEP(\mathcal{P}(M))$
but not necessarily vice versa as the following example shows:

\begin{example}
	Consider the language $L =  \{a,b\}^* a \{a,b\}^* \subseteq \{a,b\}^*$ of all words which contain the symbol $a$.
	Its syntactic monoid is the two-element monoid $M = \{1,e\}$ where $e$ is an idempotent element.
	We have $L = h^{-1}(\{e\})$ for the homomorphism $h: \{a,b\}^* \to M$ defined by $h(a) = e$, $h(b) = 1$.
	One can decide $\CEP(\mathcal{P}(M),\{e\})$ in $\NL$: For a circuit $\C$ we have
	$e \in [\C]$ if and only if an input gate with $e$ on the right-hand side is reachable from the output gate.
	However, since $M$ is not a local group, $\CEP(\mathcal{P}(M))$ is $\Ptime$-complete by Theorem~\ref{thm:power-semiring}.
	This can be also seen directly: The sets $\{e\}$ and $\{1,e\}$ form a Boolean semiring. On the other hand, 
	for the purpose of deciding $\CEP(\mathcal{P}(M),\{e\})$ one does not have to distinguish the sets
	$\{e\}$ and $\{1,e\}$. Identifying these two sets in $\mathcal{P}(M)$ yields a  $\{0,1\}$-free semiring
	whose circuit evaluation problem is in $\NL$. 
\end{example}
The example above motivates to define a congruence relation on $\mathcal{P}(M)$ where congruent subsets
are either both disjoint from $F$ or both not.
Define the equivalence relation $\sim_F$ on $\mathcal{P}(M)$ by
\[
	A_1 \sim_F A_2 \iff \forall \ell, r \in M \colon \ell A_1 r \cap F \neq \emptyset \iff \ell A_2 r \cap F \neq \emptyset 
\]
for subsets $A_1,A_2 \in \mathcal{P}(M)$. The following lemmata summarize the basic properties of $\sim_F$.

\begin{lemma}
	\label{lem:congruence}
	The following properties hold.
	\begin{enumerate}[(1)]
		\item $A_1 \sim_F A_2$ implies $(LA_1R \cap F \neq \emptyset \iff LA_2R \cap F \neq \emptyset)$ for all $L,R \subseteq M$.
		\item The relation $\sim_F$ is a congruence relation. In particular, $\mathcal{P}(M)/{\sim_F}$ is a semiring.
		\item Every $\sim_F$-class contains a largest subset with respect to  $\subseteq$.
	\end{enumerate}
\end{lemma}

\begin{proof}
	Property (1) is clear because $LA_iR \cap F \neq \emptyset$ if and only if $\ell A_i r \cap F \neq \emptyset$
		for some $\ell \in L, r \in R$.		
	For (2), assume $A_1 \sim_F A_2$ and $B_1 \sim_F B_2$. Then for all $\ell, r \in M$ we have
		\begin{align*}
		\ell (A_1 \cup B_1) r \cap F \neq \emptyset & \iff (\ell A_1 r \cup \ell B_1 r) \cap F \neq \emptyset \\
		& \iff \ell A_1 r \cap F \neq \emptyset \quad \text{or} \quad \ell B_1 r \cap F \neq \emptyset \\
		& \iff \ell A_2 r \cap F \neq \emptyset \quad \text{or} \quad \ell B_2 r \cap F \neq \emptyset \\
		& \iff (\ell A_2 r \cup \ell B_2 r) \cap F \neq \emptyset \\
		& \iff \ell (A_2 \cup B_2) r \cap F \neq \emptyset
		\end{align*}
		and, by (1),
		\begin{align*}
		\ell A_1 (B_1 r) \cap F \neq \emptyset \iff (\ell A_2) B_1 r \cap F \neq \emptyset \iff \ell A_2 B_2 r \cap F \neq \emptyset .
		\end{align*}	
	For (3) note that by (2)l, $A_1 \sim_F A_2$ implies $A_1 = A_1 \cup A_1 \sim_F A_1 \cup A_2$. Thus, every $\sim_F$-class 
	is closed under union and therefore has a largest element with respect to $\subseteq$.
\end{proof}

\begin{lemma} \label{lemma-CEP-F-quotient}
	$\CEP(\mathcal{P}(M),F)$ is  equivalent to $\CEP(\mathcal{P}(M)/{\sim_F})$ with respect to constant depth Turing-reductions.
	In other words: $\AC^0(\CEP(\mathcal{P}(M),F)) = \AC^0(\CEP(\mathcal{P}(M)/{\sim_F}))$.
\end{lemma}

\begin{proof}
	Clearly, every circuit $\C$ over $\mathcal{P}(M)$ can be regarded as a circuit $\C'$ over $\mathcal{P}(M)/{\sim_F}$
	such that $[\C']$ is the $\sim_F$-class of $[\C]$. Every $\sim_F$-class either contains only subsets of $M$ which
	are disjoint to $F$ or only subsets with non-empty intersection with $F$. Thus, $[\C']$ determines whether
	$[\C] \cap F \neq \emptyset$.
	
	For the other direction, given a circuit $\C'$ over $\mathcal{P}(M)/{\sim_F}$, we define a circuit $\C$
	over $\mathcal{P}(M)$ by picking arbitrary representative elements (subsets of $M$) for the input values 
	(which are $\sim_F$-classes) of the circuit $\C'$.
	Then we test for all $\ell,r \in M$ whether $\ell [\C] r \cap F \neq \emptyset$.
	This information is independent from the choice of representative elements and uniquely determines the $\sim_F$-class $[\C']$.
\end{proof}
From Corollary~\ref{coro-main}
and Lemma~\ref{lemma-PCFG-CEP} and \ref{lemma-CEP-F-quotient} we obtain:

\begin{theorem} \label{thm-dichotomy-PCFG}
$\mathsf{PCFG}$-$\mathsf{IP}(L,\Sigma)$ is equivalent to $\CEP(\mathcal{P}(M)/{\sim_F})$ with respect to constant depth Turing-reductions.
Therefore,
\begin{itemize}
\item $\mathsf{PCFG}$-$\mathsf{IP}(L,\Sigma)$ is  $\Ptime$-complete if $\mathcal{P}(M)/{\sim_F}$ is not $\{0,1\}$-free or its multiplicative
semigroup is not solvable,
\item $\mathsf{PCFG}$-$\mathsf{IP}(L,\Sigma)$ is in $\DET$ if $\mathcal{P}(M)/{\sim_F}$ is $\{0,1\}$-free and its multiplicative
semigroup is solvable, and
\item $\mathsf{PCFG}$-$\mathsf{IP}(L,\Sigma)$ is in $\NL$ if $\mathcal{P}(M)/{\sim_F}$ is $\{0,1\}$-free and its multiplicative
semigroup is aperiodic.
\end{itemize}
\end{theorem}

It would be nice to have a simple characterization of when $\mathcal{P}(M)/{\sim_F}$ is $\{0,1\}$-free (resp., its
multiplicative semigroup is solvable). For $\{0,1\}$-freeness, we can show:

\begin{proposition}
	\label{prop-char-0-1-free}
	$\mathcal{P}(M)/{\sim_F}$ is $\{0,1\}$-free if and only if
	\begin{equation} \label{implication-F}
        \forall s,t \in M, \, e \in E(M): st \in F \implies set \in F .
         \end{equation}
\end{proposition}

\begin{proof}
	Assume first  that $\mathcal{P}(M)/{\sim_F}$ is $\{0,1\}$-free.
	Let $s,t \in M$ such that $st \in F$ and let $e \in E(M)$.
	We have $\{e\} \sim_F \{1,e\}$ since otherwise their two $\sim_F$-classes would form a Boolean 
	subsemiring $\B_2$ in $\mathcal{P}(M)/{\sim_F}$.
	From $st \in F$ it follows that $s \{1,e\} t \cap F \neq \emptyset$ and hence
	$s \{e\} t \cap F \neq \emptyset$. Therefore $set \in F$.
	
	Assume now that the implication \eqref{implication-F} holds and towards a contradiction assume that $R$ is a subsemiring
	of $\mathcal{P}(M)$ with the zero-element $[A]_{\sim_F}$ and the one-element $[B]_{\sim_F}$.
	We choose $A,B$ to be the $\subseteq$-maximal elements in their classes.
	Then we have $A \subseteq B$ because $[A]_{\sim_F} \cup [B]_{\sim_F} = [B]_{\sim_F}$.
	Further $A^2 \sim_F A$ implies $A^2 \subseteq A$ by maximality of $A$ and therefore
	$A$ contains at least an idempotent $e \in A$ (take $a^\omega$ for any $a \in A$).
	Finally we have $AB \sim_F A$, which implies $AB \subseteq A$.
	
	Since $A$ and $B$ are not $\sim_F$-equivalent there exist elements $s,t \in M$ such that
	$sAt \cap F = \emptyset$ but $sbt \in F$ for some $b \in B$.
	However, $eb \in AB \subseteq A$ and by assumption $sebt \in F$, contradiction.
\end{proof}
We do not have a nice characterization for solvability of the multiplicative semigroup of $\mathcal{P}(M)/{\sim_F}$.

Let us conclude this section with an application of Corollary~\ref{thm-dichotomy-PCFG}:
\begin{example}
	Consider a language of the form $L = \Sigma^* a_1 \Sigma^* a_2 \Sigma^* \dots a_k \Sigma^* $ for $a_1, \dots, a_k \in \Sigma$,
	which is a so called {\em piecewise testable language}.
	We claim that $\mathsf{PCFG}$-$\mathsf{IP}(L,\Sigma)$ is decidable in $\NL$.
	First, since $uw \in L$ implies $uvw \in L$ for all $u,v,w \in \Sigma^*$, the syntactic monoid $M$ and the accepting subset $F \subseteq M$ of $L$ clearly satisfies
	the condition of Proposition~\ref{prop-char-0-1-free}.
	Second, clearly $\mathcal{P}(M)_{+}$ is aperiodic and hence also $(\mathcal{P}(M)/{\sim_F})_{+}$.
	Third, we show that $\mathcal{P}(M)_{{\scriptscriptstyle \bullet}}$ and hence also $(\mathcal{P}(M)/{\sim_F})_{{\scriptscriptstyle \bullet}}$ is aperiodic.
	Simon's theorem \cite{Simon75} states that a language is piecewise testable if and only if its syntactic monoid is $\mathcal{J}$-trivial.
	We claim that $\mathcal{P}(M)_{{\scriptscriptstyle \bullet}}$ is also $\mathcal{J}$-trivial, in particular aperiodic.
	Let $A, B \in \mathcal{P}(M)$ such that $A \equiv_\mathcal{J} B$.
	Consider the directed bipartite graph on $A \uplus B$ with edges
	\[
	\{ (a,b) \in A \times B \mid a \geq_\mathcal{J} b \} \cup \{ (b,a) \in B \times A \mid b \geq_\mathcal{J} a \}.
	\]
	Every vertex has at least one outgoing and one incoming edge, which means that it belongs to a non-trivial strongly connected component.
	Since $M$ is $\mathcal{J}$-trivial, we must have $A = B$.
\end{example}

\section{Some results about infinite semirings}

It would be interesting to see, whether our techniques can be extended to certain infinite semirings. Recently, it was shown that for certain finitely generated (but infinite) linear groups,
the circuit evaluation problem belongs to $\NC^2$ or at least $\coRNC^2$ (the complement of the randomized version of $\NC^2$) \cite{KonigL15,KonigL15mfcs}.
Of course, if one deals with infinite structures, one needs a finite representation of the elements of the structure. Moreover, in contrast to finite structures, the circuit 
evaluation problem (i.e., the question, whether a given circuit evaluates to a given element) is not equivalent to its computation variant, where one wants to compute
the output value of the circuit.  A good example is the arithmetic ring $(\mathbb{Z},+,\cdot)$. Whether a given circuit over this ring evaluates to a given element is 
equivalent to the question, whether a given circuit evaluates to zero, which, in turn, is equivalent to the famous polynomial identity testing problem \cite{AllenderBKM09}. Its complexity
is in co-randomized polynomial time and no deterministic polynomial time algorithm is known (the problem is easily seen to be $\Ptime$-hard). On the other hand, by iterated squaring one can construct a circuit
over $(\mathbb{Z},+,\cdot)$ with $n$ gates that evaluates to $2^{2^n}$. The binary representation of this number needs $2^n$ bits. This shows that in general,
we cannot even write down the output value of a given circuit in polynomial time. But even for semirings, where this phenomenon does not occur, it seems to be difficult
to obtain $\NC$-algorithms.  As an example, let us consider one of the simplest infinite semirings, namely the max-plus semiring
$(\mathbb{N}, \max, +)$. Note that it is $\{0,1\}$-free (this follows from Lemma~\ref{lemma-0-1-free}) and $+$ (the semiring multiplication) 
is commutative. 

\begin{theorem}
$\CEP(\mathbb{N}, \max, +)$ is $\Ptime$-complete.
\end{theorem}

\begin{proof}
A given circuit over $(\mathbb{N}, \max, +)$ can be sequentially evaluated in polynomial time by representing 
all numbers in binary notation.
Note that a $\max$-gate does not increase the number of bits, whereas a $+$-gate can increase the number of bits
by at most one (the sum of an $n$-bit number and an $m$-bit number has at most $\max\{n,m\}+1$ many bits).

For the lower bound we reduce from $\CEP(\mathbb{B}_2)$.
Let $\C = (V,A_0,\rhs_{\C})$ be a circuit over the boolean semiring.  W.l.o.g. we can assume that
$\C$ consists of $n$ layers, where all wires go from layer $k$ to layer $k+1$ for some $k$.
Layer 1 contains the input gates and layer $n$ contains the output gate $A_0$. 
We now construct a circuit $\D = (V,A_0,\rhs_{\D})$ over $(\mathbb{N}, \max, +)$ with the same gates as $\C$.
The idea is to construct $\D$ such that the following conditions hold for every gate $A \in V$ on layer $k$:
\begin{enumerate}[(a)]
\item If $[A]_{\C} = 0$ then $[A]_{\D} = 2^k-1$
\item If $[A]_{\C} = 1$ then $[A]_{\D} = 2^k$
\end{enumerate}
To get this correspondence, we define the right-hand sides for $\D$ as follows,
where $B$ and $C$ are on layer $k < n$ and $A$ is on layer $k+1$.
\begin{itemize}
\item If $\rhs_{\C}(A) = B \wedge C$ then $\rhs_{\D}(A) = \max( B+C, 2^{k+1}-1 )$.
\item If $\rhs_{\C}(A) = B \vee C$ then $\rhs_{\D}(A) = \max(B,D)+2^{k}$.
\end{itemize}
Moreover, if $\rhs_{\C}(A) = 0$ (resp., $\rhs_{\C}(A)=1$), then gate $A$ is on layer $1$ and 
we set $\rhs_{\D}(A) = 1$ (resp., $\rhs_{\D}(A) = 2$).
With these settings, it is straightforward to show that (a) and (b) hold. In particular, we have $[\C] = 1$ if and only if $[\D] = 2^n$.
\end{proof}
Let us conclude this section with a few results on power semirings, in particular, power semirings of groups. 
Recall from Example~\ref{ex-power-semigroup} that for a finite solvable group $G$, the semiring of non-empty subsets $\mathcal{P}(G)$ has a circuit evaluation
problem in $\DET$. This motivates the question for the complexity of  the circuit evaluation problem of $\mathcal{P}(G)$ for an infinite but finitely generated group $G$.
To avoid the problem of representing the output value of a circuit over $\mathcal{P}(G)$ (which may be a set of exponentially many group elements)
we consider again the variant $\CEP(\mathcal{P}(G),F)$ from Section~\ref{sec-intersection},
which asks whether the set computed by a given circuit over $\mathcal{P}(G)$ contains an element from a given set $F \subseteq G$.
We can assume that
the input gates of the circuit are labelled with singleton sets $\{a\}$, where $a$ is a generator of $G$. 
This variant of the circuit evaluation problem can be seen as a nondeterministic
version of the compressed word problem for $G$ (i.e., $\CEP(G)$) \cite{Loh14}.
As an example, a result from \cite{StMe73} can be reformulated as follows: $\CEP(\mathcal{P}(\mathbb{Z}),\{0\})$ is $\NP$-complete,
where $\mathbb{Z}$ is the additive group of integers, i.e., the free group of rank $1$. For the free group of rank $2$, briefly $F_2$,
we have: 

\begin{theorem}
$\CEP(\mathcal{P}(F_2),\{1\})$ is $\PSPACE$-complete.
\end{theorem}

\begin{proof}
Let us fix the generating set $\{a,a^{-1},b,b^{-1}\}$ for $F_2$.
We first show that $\CEP(\mathcal{P}(F_2),\{1\})$ is in $\PSPACE$. 
First, we show that $\CEP(\mathcal{P}(F_2),\{1\})$ restricted
to circuits that are trees is in $\LOGCFL$, i.e., it can be solved in polynomial time on a nondeterministic pushdown machine with an auxiliary
working tape of logarithmic length. For this, the machine traverses the input circuit (a tree) in a depth-first left-to-right manner and thereby
stores the current node in the tree. Each time, the machine arrives at a $\cup$-gate $A$ from $A$'s parent gate, it nondeterministically
decides whether it proceeds to the left or right child of $A$. If it goes to the right child, then the subtree rooted in the left child of $A$ is omitted in the traversal.
Similarly, if the machine goes to the left child, then when returning to $A$ it omits the right subtree of $A$
in the traversal. Thus, the machine chooses nondeterministically for each union gate one of the two children and only traverses the subtree of the 
chosen child. Finally, when the machine arrives at a leaf of the tree, i.e., an input gate of the circuit, and this leaf is labelled with the group generator 
$x \in \{a,a^{-1},b,b^{-1}\}$ then it pushes $x$ on the pushdown. If the top of the pushdown then ends with $x^{-1} x$, it pops $x^{-1} x$ from the pushdown.
Then the machine continues with the traversal of the tree. At the end of the traversal, the machine accepts if and only if the pushdown is empty.

Basically, the above machine nondeterministically chooses a word over $\{a,a^{-1},b,b^{-1}\}$ that belongs to the circuit output if the circuit
is interpreted over the power semiring $\mathcal{P}(\{a,a^{-1},b,b^{-1}\}^*)$. Moreover, while generating this word, it verifies, using the pushdown,
that the word is equal to $1$ in the free group $F_2$.
Now we use the fact that $\LOGCFL$ is contained in $\DSPACE(\log^2(n))$. Hence, $\CEP(\mathcal{P}(F_2),\{1\})$ restricted
to circuits that are trees is in $\DSPACE(\log^2(n))$. From this, it follows easily that $\CEP(\mathcal{P}(F_2),\{1\})$ for general (non-tree-like)
circuits belongs to $\PSPACE$: The function that maps a circuit $\C$ to its unfolding (a tree that is equivalent to $\C$) can be computed
by a Turing machine with output in polynomial space (a so called $\PSPACE$-transducer). 
This follows from the fact that the gates of the unfolding of $\C$ can be identified with
paths in $\C$ that start in the output gate and go down in the circuit. The set of all these paths can be produced in polynomial space. 
Now we use a simple lemma (see \cite[Lemma~1]{LohreyM13} stating that a preimage $f^{-1}(L)$ belongs to $\PSPACE$, if $f$ can be computed
by a $\PSPACE$-transducer and the language $L$ can be decided in polylogarithmic space.

For the lower bound we use a result from  \cite{NeSa04}, saying that the intersection non-emptiness problem for 
given {\em acyclic} context-free grammars $\mathcal{G}_1$ and $\mathcal{G}_2$ is $\PSPACE$-complete. 
By coding terminal symbols into a binary alphabet, we can assume that the terminal alphabet of $\mathcal{G}_1$ and $\mathcal{G}_2$
is $\{a,b\}$. Moreover, by the construction in \cite{NeSa04}, $\mathcal{G}_1$ and $\mathcal{G}_2$ are productive 
(which for acyclic grammars just means that every nonterminal has a rule).
For a word $w = a_1 a_2 \cdots a_n$ with $a_i \in \{a,b\}$ let $w^{-1} = a_n^{-1} \cdots a_2^{-1} a_1^{-1}$. 
It is straightforward to construct from  $\mathcal{G}_1$ and $\mathcal{G}_2$ an acyclic productive context-free grammar $\mathcal{G}$ over the terminal alphabet
$\{a,b,a^{-1},b^{-1}\}$ 
such that
$$
L(\mathcal{G}) = \{ w_1 w_2^{-1} \mid w_1 \in L(\mathcal{G}_1), w_2 \in L(\mathcal{G}_2) \}.
$$
Now an acyclic productive context-free grammar with terminal alphabet $\Sigma$ can be seen as a circuit over the semiring 
$\mathcal{P}(\Sigma^*)$. By interpreting the circuit for the grammar $\mathcal{G}$ 
as a circuit $\C$ over $\mathcal{P}(F_2)$ we see that 
$L(\mathcal{G}_1) \cap L(\mathcal{G}_2) \neq \emptyset$ if and only if $[\C]$ contains $1$.
\end{proof}

\section{Conclusion and outlook}

We proved a dichotomy result for the circuit evaluation problem for finite semirings: If (i) the semiring has no subsemiring with an additive and 
multiplicative identity and both are different and (ii) the multiplicative subsemigroup is solvable, then the circuit evaluation problem is in $\DET \subseteq \NC^2$, otherwise it is $\Ptime$-complete. 

The ultimate goal would be to obtain such a dichotomy
for all finite algebraic structures. One might ask whether for every finite algebraic structure $\mathcal{A}$, 
$\CEP(\mathcal{A})$ is $\Ptime$-complete or in $\NC$. It is known that under the assumption $\Ptime \neq \NC$ there exist
problems in $\Ptime \setminus \NC$ that are not $\Ptime$-complete \cite{Vollmer90}. In \cite{BeMc97} it is shown that every circuit evaluation
problem $\CEP(\mathcal{A})$ is equivalent to a circuit evaluation
problem $\CEP(A, \circ)$, where $\circ$ is a binary operation.

\medskip
\noindent
{\bf Acknowledgment.} We thank Benjamin Steinberg for helpful discussions on semigroups. We are grateful to Volker Diekert for pointing out to us
the proof of the implication ($3 \Rightarrow 4$) in the proof of Lemma~\ref{lemma-0-1-free}.

%\bibliography{bib}

\end{document}